\documentclass[11pt]{article}
\usepackage{latexsym}
\usepackage{physics}
\usepackage{hyperref}
\usepackage{tikz}

\hypersetup{hypertex=true,
	colorlinks=true,
	linkcolor=blue,
	anchorcolor=blue,
	citecolor=blue}
\usepackage{amssymb}
\usepackage{amsmath}
\usepackage{color}
\usepackage{amsthm}

\usepackage{graphicx}
\usepackage{cleveref}
\crefformat{section}{\S#2#1#3} 
\crefformat{subsection}{\S#2#1#3}
\crefformat{subsubsection}{\S#2#1#3}
\usepackage{enumerate}
\usepackage{tikz}
\theoremstyle{plain}
\newtheorem{thm}{Theorem}[section]
\newtheorem{lem}{Lemma}[section]
\newtheorem{cor}[thm]{Corollary}
\newtheorem{prop}[lem]{Proposition}

\newtheorem{prop-defn}[lem]{Proposition/Definition}
\newtheorem{rem}[lem]{Remark}

\newtheorem{defn}[lem]{Definition}

\newcommand{\C}{\mathbb{C}}

\newcommand{\R}{\mathbb{R}}
\newcommand{\Q}{\mathcal{Q}}
\def\P{\mathcal{P}}
\newcommand{\Z}{\mathbb{Z}}
\newcommand{\T}{\mathcal{T}}

\renewcommand{\H}{\mathcal{H}}
\def\B{\mathcal{B}}

\newcommand{\F}{\mathbf{F}}

\DeclareMathOperator{\End}{End}

\DeclareMathOperator{\Hom}{Hom}

\DeclareMathOperator{\GM}{GM}

\DeclareMathOperator{\prim}{prim}

\DeclareMathOperator{\Jac}{Jac}
\DeclareMathOperator{\ord}{ord}
\DeclareMathOperator{\res}{res}

\DeclareMathOperator{\Id}{Id}

\DeclareMathOperator{\w}{wt}

\DeclareMathOperator{\Span}{Span}

\newcommand{\be}{\begin{equation}}
\newcommand{\ee}{\end{equation}}
\newcommand{\bc}{\begin{cases}}
	\newcommand{\ec}{\end{cases}}
\newcommand{\bes}{\begin{equation*}}
\newcommand{\ees}{\end{equation*}}
\newcommand{\ba}{\begin{align}}
\newcommand{\ea}{\end{align}}
\newcommand{\bas}{\begin{align*}}
\newcommand{\eas}{\end{align*}}
\newcommand{\es}{\end{split}}
\newcommand{\bs}{\begin{split}}
\newcommand{\A}{\mathcal A}

\renewcommand{\H}{\mathcal H}

\def\p{\partial}
\def\bap{\bar{\partial}}
\def\I{\mathcal{I}}
\begin{document}
	\title{CY/LG Correspondence for Weil-Petersson Metrics and $tt^*$ Structures}
	
	\date{}                                      

	\author{Xinxing Tang
		\footnote{Yanqi Lake Beijing Institute of Mathematical Sciences and Applications, Beijing, China, tangxinxing@bimsa.cn. Partially supported by Tsinghua Postdoc Grant 100410058 and research funding at BIMSA.
		}
		\and Junrong Yan
		\footnote{Beijing International Center for Mathematical Research, Peking University, Beijing, China 100871, j\_yan@bicmr.pku.edu.cn. Partially supported by Boya Postdoctoral Fellowship at Peking University.
		} 
	}
	
	\maketitle
	
	\abstract{The aim of this paper is to rigorously establish the Calabi-Yau/Landau-Ginzburg (CY/LG) correspondence for the $tt^*$ geometry structure--a generalized version of variation of Hodge structures. Although it is well-known that there exists a map between Hodge structures on the LG and CY's sides that preserves the Hodge filtration and bilinear form, it remains unclear whether the real structures are also preserved. In our paper, we conduct a detailed analysis of two period integrals on the LG's side. Based on this analysis, we modify the real structure (c.f. \cite[(4.2)]{CECOTTI1991N}) proposed by Cecotti on LG's side, and show that the aforementioned map is also preserved under the modified real structure (see \cref{last}). As a result, we establish full CY/LG correspondence for $tt^*$ structures.}
	\tableofcontents
	\section{Introduction}
	\subsection{Overview}
	

	In the 1990s, physicists discovered mirror symmetry, a relationship between the symplectic geometry (A-model) of a Calabi-Yau manifold $X$ and the complex geometry (B-model) of its mirror $X^{\vee}$. Since its formulation over three decades ago \cite{Candelas1991APO,Greene1990DualityI}, the mirror symmetry conjecture has had a profound influence on mathematics. The A-model on a Calabi-Yau manifold $X$ is the famous Gromov–Witten theory of $X$. The genus 0 theory of the Calabi-Yau B-model is related to the variation of Hodge structure, whereas the higher genus B-model theory is known as BCOV theory \cite{Bershadsky1994KodairaSpencerTO}, \cite{CostLi}.
	
	Simultaneously, physicists extended the above discussion to singularity theory, which became known as the Landau-Ginzburg model. A Landau-Ginzburg model is defined on the pair $(X,f)$, where $X$ is a complete noncompact  K\"{a}hler manifold and $f$ is a holomorphic function. For the A-side, the Landau-Ginzburg A-model is constructed by Fan-Jarvis-Ruan \cite{FJR2013} following Witten's proposal \cite{Witten1993}, which is known as FJRW-theory. Moreover, Chang-Li-Li \cite{chang2015witten} recently formulated FJRW-theory in an algebraic geometric formulation. On the B-side, Saito's theories of primitive form  \cite{1981Primitive} and  the higher residue pairing \cite{Saitohigher} generate the structure of a Frobenius manifold on the singularity's universal deformation space \cite{SaitoFM}, which yield the  genus-0 theory of Landau-Ginzburg B-model. Also, some Hodge theoretical aspect in LG B-model, called the $tt^*$-geometry, is discovered by Cecotti-Vafa \cite{cecotti1991topological}, whose integrability structure is studied by Dubrovin \cite{Dub1993tt}. Later, Hertling \cite{Hert2003tt} carefully researched the $tt^*$ geometric structure and organized a variety of known structures into the so-called TERP structure. Additionally, Fan \cite{fan2011schr} introduced an analytic approach in the spirit of $N=2$ geometry by studying the spectral theory of a twisted Laplacian operator. Besides, Li-Wen \cite{li2019l2} used the $L^2$ Hodge theory to give a Frobenius manifold structure for the case $f$ with compact critical locus.
	Furthermore, motivated by the Virasoro equations and localization calculation in the A-model at a higher genus, Givental \cite{Giv2001b} gave a remarkable formula for the partition function in the semi-simple case.

	The Calabi-Yau/Landau-Ginzburg correspondence connects nonlinear sigma models on Calabi-Yau manifolds with Landau-Ginzburg models. It turns out that CY/LG correspondence and mirror symmetry have served as guidelines in the study of many branched of mathematics (see the following diagram):
	
	
	\vskip 0.2cm
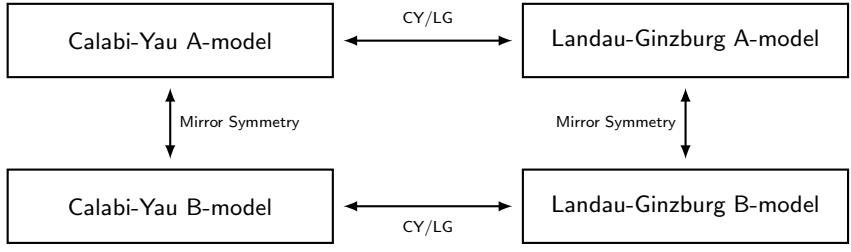
	\begin{center}
		\sffamily
		\footnotesize
		\begin{tikzpicture}[auto,
		block_center/.style ={rectangle, draw=black, thick, fill=white,
			text width=12em, text centered,
			minimum height=3em,inner sep=6pt,outer sep=4pt},
		block_left/.style ={rectangle, draw=black, thick, fill=white,
			text width=12em, text centered, minimum height=3em, inner sep=6pt,outer sep=4pt}]
		\matrix [column sep=25mm,row sep=12mm] {
			\node [block_left] (CYA) {Calabi-Yau A-model};
			& \node [block_center] (LGA) {Landau-Ginzburg A-model};\\
			\node [block_left] (CYB) {Calabi-Yau B-model};
			& \node [block_center] (LGB) {Landau-Ginzburg B-model};\\
		};
		\path[draw=black,latex-latex,thick] (CYA)  -- node[above,font=\tiny,outer sep=2pt]{CY/LG}(LGA);
		\path[draw=black,latex-latex,thick] (LGA) -- node[left, font=\tiny, align=left]{Mirror Symmetry}(LGB);
		\path[draw=black,latex-latex,thick] (CYB) -- node[below,font=\tiny,outer sep=2pt] {CY/LG} (LGB);
		\path[draw=black,latex-latex,thick] (CYA) -- node[right,font=\tiny,align=left] {Mirror Symmetry} (CYB);
		\end{tikzpicture}
	\end{center}
	The CY/LG correspondence for the $A$ model is partially done by Chiodo-Iritani-Ruan \cite{chiodo2014landau}. The question of whether the CY/LG correspondence holds for the B-model is both intriguing and difficult.  To address this question for the genus 0 cases, we show CY/LG correspondence for $tt^*$ structures. In fact, Carlson-Griffiths  \cite{carlson1980infinitesimal} compared the cup product and symplectic pairing on the CY and LG sides. Also, Cecotti \cite{CECOTTI1991N} investigated the $tt^*$ geometry structure on both sides from a physical standpoint.  
	Besides, Fan-Lan-Yang \cite{Fan2020LGCYCB} partially proved that the two $tt^*$ structures are isomorphic except for the real structures, and gave another proof of full CY/LG correspondence of $tt^*$ structures \cite{fan2022constructing} (See Remark \ref{diff} for more details). Here we use different methods than \cite{Fan2020LGCYCB,fan2022constructing} to show the full CY/LG correspondence: we introduce two $U(1)$ actions (see \cref{2.4}), which act as certain bi-grading for LG B-models.  We also establish CY/LG correspondence for Weil-Petersson-type metrics using these two $U(1)$ actions. Furthermore, the Agmon estimate derived in \cite{DY2020cohomology} plays an essential role in our method (See also \cref{witagm}). More precisely, we can use the Agmon estimate to compute  the transition matrices (c.f. \cite[(A.9)]{cecotti1991topological}) introduced by Cecotti-Vafa more explicitly (Theorem \ref{lgcy1}). 
	
	Lastly, we would like to emphasize the following: In \cite{CECOTTI1991N}, Cecotti also introduces a real structure on $\Jac(f)$ using oscillation integrals (c.f. \cite[(4.15)]{CECOTTI1991N}, so Cecotti introduced two real structures in his paper, the other is \cite[(4.2)]{CECOTTI1991N}) and subsequently partially shows CY/LG correspondences for $tt^*$ structures in the semi-infinite sense (see also \cite{iritani2021gromov}). While in this paper, we show CY/LG correspondence for $tt^*$ structures without passing to the semi-infinity setting.
	
	\subsection{$tt^*$ geometry structure on Calabi-Yau manifolds}
	
	Let $X$ be a compact Calabi-Yau $(n-2)$-fold and $\mathcal{M}$ be the moduli stack of the complex structure of $X$ with the universal family
	$\pi: \mathcal{X}\rightarrow\mathcal{M}$.

	Let $H_{\prim}^{n-2}$ be the Hodge bundle of primitive forms over $\mathcal{M}$, i.e., for any $[X]\in\mathcal{M}$, the fiber $H_{\prim}^{n-2}\big|_{[X]}$ is the space of primitive $(n-2)$-forms on $X$. $H_{\prim}^{n-2}$ is called state space by physists. 

We have the Hodge decomposition and Hodge filtration:
	\begin{align*}
	&H^{n-2}_{\prim}=H^{n-2,0}_{\prim}\oplus H^{n-3,1}_{\prim}\oplus\cdots\oplus H^{0,n-2}_{\prim},\quad \overline{H^{n-2-q,q}_{\prim}}=H^{q,n-2-q}_{\prim};\\
	&\F^{n-2}H^{n-2}_{\prim}\subset\cdots\subset \F^0H^{n-2}_{\prim}=H^{n-2}_{\prim},\quad \F^pH^{n-2}_{\prim}=\oplus_{p'=p}^{n-2}H_{\prim}^{p',n-2-p'}.
	\end{align*}
	Moreover, such a bundle is equipped with
	\begin{itemize}
		\item a flat Gauss-Manin connection $\nabla^{\GM}$ which satisfies the Griffiths transversality conditions:
		$$\nabla^{\GM}_{v}\F^p\subset \F^{p-1}, \quad v\in \mathcal{T}_{\mathcal{M}}.$$
		\item a pairing $\eta: H^{n-2}_{\prim}\otimes H^{n-2}_{\prim}\rightarrow C^{\infty}(\mathcal{M})$ such that
		$$g(\alpha,\beta):=i^{2p-(n-2)}\eta(\alpha,\bar\beta)=i^{2p-(n-2)}\int_{X}\alpha\wedge\bar{\beta},\quad \alpha,\beta\in H^{p,n-2-p}_{\prim}$$
		gives a positive Hermitian metric on $H^{n-2}_{\prim}$.
		We call $g$ the $tt^*$ metric in the CY B-model.
	\end{itemize}

	In particular, $H^{n-2,0}_{\prim}$ forms a line bundle over $\mathcal{M}$. Let $\Omega_u$ be a local holomorphic section of $H^{n-2,0}_{\prim}$, which, fiberwisely, is a holomorphic volume form on $X$. The K\"{a}hler potential $K(u,\bar{u})$ is defined via
	\be\label{kpo}e^{-K(u,\bar{u})}=i^{-n-2}\int_{X}\Omega_u\wedge\overline{\Omega_u}.\ee
	Moreover, the Weil-Petersson metric is defined by
	$$G^{CY}=\partial\bar{\partial}K(u,\bar{u}).$$
	It can be checked easily that $G^{CY}$ is independent of the choice of the local holomorphic section of $H^{n-2,0}_{\prim}$.
	
	The variation of Hodge structure is related to the so-called $tt^*$ geometry structure:
	\def\l{\lambda}
	
	\begin{defn}[$tt^*$ geometry] A $tt^*$ geometry structure $(K\rightarrow M, \kappa, g, D, \bar{D}, C, \bar{C})$ consists of the following data
		\begin{itemize}
			\item $M$ is a complex manifold, $K\rightarrow M$ is a holomorphic vector bundle,
			\item a complex anti-linear involution $\kappa: K\rightarrow K$, i.e. $\kappa^2=\Id$, $\kappa(\l\alpha)=\bar{\l}\kappa(\alpha)$, $\forall~ \l\in\C,\alpha\in\Gamma(K)$,
			\item a Hermitian metric $g$ on $K$,
			\item a one-parameter family of flat connections $\nabla^z=D+\bar{D}+\frac{1}{z}C+z\bar{C}$, where $z\in\C^*$, $D$ and $\bar{D}$ are the $(1,0)$ and $(0,1)$ part of the Chern connection of $g$ respectively, $C$, $\bar{C}$ are $C^{\infty}(M)$-linear maps
			$$C:\Gamma(K)\rightarrow \Gamma(K)\otimes_{C^{\infty}(M)}\mathcal{A}^{1,0}(M),\quad \bar{C}:\Gamma(K)\rightarrow \Gamma(K)\otimes_{C^{\infty}(M)}\mathcal{A}^{0,1}(M),$$
		\end{itemize}
		satisfying
		\begin{enumerate}
			\item $g$ is real with respect to $\kappa: g(\kappa(w_1),\kappa(w_2))=\overline{g(w_1,w_2)}$, $\forall w_1,w_2\in\Gamma(K)$,
			\item $(D+\bar{D})(\kappa)=0$, $\bar{C}=\kappa\circ C\circ\kappa$,
			\item $\bar{C}$ is the adjoint of $C$ with respect to $g$, i.e. $g(C_Xw_1,w_2)=g(w_1,\bar{C}_{\bar{X}}w_2)$, $X\in\mathcal{T}_M$.
		\end{enumerate}
		Then we say such a structure $(K\rightarrow M, \kappa, g, D, C, \bar{C})$ is a $tt^*$ geometry structure.
	\end{defn}

	Now consider the  $tt^*$ geometry structure induced by the variation of the Hodge structure on the Calabi-Yau B model. To begin, the $tt^*$ connection is the Chern connection of the $tt^*$ metric; denote the $(1,0)$ and $(0,1)$ parts of the $tt^*$ connection by $D^{CY}$ and $\bar{D}^{CY}$, respectively.
	
	Following that, let $$C^{CY}:T_{\mathcal{M}}\longrightarrow \bigoplus_p\Hom(H^{p,n-2-p}_{\prim},H^{p-1,n-2-p+1}_{\prim})\subset\End(H^{n-2}_{\prim})
	$$
	be the Kodaira-Spencer map and $\bar{C}^{CY}$ be its complex conjugate.

	Combined with the complex conjugation and the $tt^*$ metric $g$,  the Griffiths transversality implies that the above data define a $tt^*$ structure on CY's side.
	
	Let $u=(u_1,...,u_s)$ be local coordinates of the moduli $\mathcal{M}$, and $C^{CY}_i:=C^{CY}(\frac{\p}{\p u_i})$, $G^{CY}_{i\bar{j}}=\p_{u_i}\bar{\p}_{\bar{u}_j}\log K(u,\bar{u}) $, where the K\"ahler potential $K$ is defined via (\ref{kpo}).
	By Griffiths transversality, one can show that $G_{i\bar{j}}$ is related to the $tt^*$ metric as follows:
	$$G_{i\bar{j}}^{CY}=\frac{g\left(C_i^{CY}\Omega_u,C_j^{CY}\Omega_u\right)}{g(\Omega_u,\Omega_u)}.$$

	\subsection{State spaces in LG model}
	A Landau-Ginzburg model is defined by the pair $(X, f)$, where $X$ is a noncompact complete K\"ahler manifold and $f$ is a holomorphic function on $X$. The corresponding B-model is related to the deformation theory of singularities.
	
	It is well-known that the Landau-Ginzburg (LG) B-model is closely
	related to Sigma B-models on Calabi-Yau manifolds. For instance, if we consider the LG model  $(\mathbb{C}^5, f(\textbf{z})=z_1^5+\cdots+z_5^5)$, regarding $[z_1,\ldots,z_5]$ as the homogeneous coordinate system on $\mathbb{P}^4$, then $f=0$ defines a hypersurface $X_f$ (hence is a Calabi-Yau 3-fold) in $\mathbb{P}^4$. In particular, in the physical literature, the central charge $\hat{c}_f$ (defined below) of $f$ is given by $5-2=3$, which is equal to the dimension of the hypersurface. In fact, more quantum information of $X_f$ can be read from $(\mathbb{C}^5, f)$. Furthermore, consider the deformation
	$$F=f+u\prod_{i=1}^5z_i, \quad u\in\mathbb{C}.$$
	For sufficiently small $|u|$, $F$ determines a family of Calabi-Yau hypersurfaces $X_{F(\cdot,u)}$ parametrized by $u$, i.e., a deformation of $f$ induces a deformation of complex structure of $X_f$.

	 In the physical literature, this phenomenon is referred to as the Calabi-Yau/Landau-Ginzburg correspondence. Before we continue, let us briefly review some well-known results.
	


	Consider the pair $(\C^n,f)$, where $f:\C^n\rightarrow \C$ is a quasi-homogeneous polynomial with isolated singularity at the origin, i.e. there exist $q_1,\ldots,q_n\in\mathbb{Q}_+$ such that for any $\lambda\in\C^*$,
	$$f(\lambda^{q_1}z_1,\ldots, \lambda^{q_n}z_n)=\lambda f(z_1,...,z_n).$$
	Each $q_i$ is referred to as the weight of  $z_i$, denoted by the expression $\w(z_i)=q_i$. The system of equations $\partial_i f=0$ has only one solution: zero (with multiplicity).

	It is known that, as a vector space, the state space of Landau-Ginzburg model is isomorphic to the Jacobi ring $\Jac(f)$ of $f$, where
	$$\Jac(f)=\C[z_1,\ldots,z_n]/\langle\p_1f,\ldots,\p_nf\rangle\cong \Omega_{\C^n}^{n}/df\wedge\Omega_{\C^n}^{n-1}.$$
	
	According to the Griffiths-Carlson theory (see \cref{two}), for a non-degenerate homogeneous degree $n$ polynomial $f$, we consider the subspace $\Jac(f)'$ of the Jacobi ring of $f$, given by
	$$\Jac(f)'=\Span\left\{\phi_i~\bigg|~\frac{\deg\phi_i}{n}\in\mathbb{N}, \{\phi_i\}_{i=0}^{\mu-1} \text{ is a homogeneous basis of $\Jac(f)$}\right\}.$$
	Then we will establish the CY/LG correspondence between $\Jac(f)'$ and $H^{n-2}(X_f)_{\prim}$.

	\subsection{Main theorems}
	
	
	In this paper, we investigate the $tt^*$ structure of the Landau-Ginzburg B-model and CY/LG correspondence. We study the $tt^*$ geometric structure in this paper from the viewpoint of Schr\"{o}dinger representations of 1d Landau-Ginzburg model \cite{cecotti1991topological} and \cite{fan2011schr}, see next section for more details.
	
	Consider the real twisted Laplacian operator $\Delta_f$ associated to $(\C^n,f)$ (See \cref{Hodgedecomposition} for the definition of $\Delta_f$). The space of harmonic forms $H_f$ (with respect to $\Delta_f$) is called the vacuum space in the physical literature, and is isomorphic to the Jacobi ring of the polynomial $f$ via the so-called spectral flow $S$ (See the definitions in \cref{S}). The complex conjugate on differential forms gives a natural real structure on $H_f$.
	
	Furthermore, consider the marginal deformation $F$ of $f$ parametrized by the space $u\in M$
	$$F(z,u)=f(z)+\sum_{i=1}^su_i\psi_i,$$
	where $\psi_i$ has the same weight as $f$ (see Definition \ref{quasi}).
	
	The Hodge bundle is defined over the deformation space $M$. Fiberwise, at $u\in M$, it is given by the  harmonic forms of $\Delta_{F(\cdot,u)}$. Sometimes $\Delta_{F(\cdot,u)}$ is also denoted by $\Delta_F$ to avoid heavy notations. The $tt^*$ metric on $M$ is defined as
	$$h(\alpha,\beta)=\int_{\C^n}\alpha\wedge*\bar{\beta},\quad \alpha,\beta\in\ker(\Delta_{F}). $$
	Next, we construct a $C^{\infty}(M)$-linear isomorphism (see \cref{S} for details)
	$$S:C^{\infty}(M)\otimes\Omega_{\C^n\times M/M}^{n}/dF\wedge\Omega_{\C^n\times M/M}^{n-1}\longrightarrow\ker\Delta_F.
	$$
	Via the isomorphism $S$ and some $U(1)$ action analysis, we will study the vacuum line bundle and Weil-Petersson-type metric in Laundau-Ginzburg B-model. 
	We will derive the similar relation between $tt^*$ metric and Weil-Petersson metric as in Calabi-Yau's case:

	Let $u=(u_1,...,u_s)$ be local coordinates of  $M$, and $\p_i:=\partial_{u_i} $.
	
	\begin{thm} \label{131}The Weil-Petersson-type metric $G_{i\bar{j}}$ is related to the $tt^*$ metric $h$ in the following way:
		$$G_{i\bar{j}}=\frac{h(w_i,w_j)}{h(w_0,w_0)},$$
		where
		\begin{itemize}
			\item $w_0=S(dz_1\wedge\cdots\wedge dz_n)$ is a holomorphic section of the vacuum line bundle, and
			$w_i=S(\psi_idz_1\wedge\cdots\wedge dz_n)$ is the holomorphic section corresponding to the deformation $\psi_i$.
			\item $G_{i\bar{j}}:=-\partial_i\bar{\p}_{\bar{j}}\log h(w_0,w_0)$.
		\end{itemize}
	\end{thm}
	
	Then, to understand $tt^*$ structure, we investigate the following two types of period integrals (see \cref{periodintegral} for more details).
	$$\int_{\gamma^-}e^{F+\bar{F}}w\quad\text{and}\quad \int_{\gamma^-}e^FA,$$
	where $w=S(A)$.
	
	By choosing a basis $\{\gamma_k^-\}$  (Lefschetz thimble) in the homology class $H_n(\C^n, F^{-\infty};\mathbb{Z})$ and a holomorphic basis $\{A_a\}$ of $C^{\infty}(M)\otimes\Omega_{\C^n\times M/M}^n/dF\wedge\Omega_{\C^n\times M/M}^{n-1}$, there exists a matrix-valued function $\T(u,\bar{u})$ such that
	$$\int_{\gamma_k^-}e^FA_a=\sum_b \T_{ab}(u,\bar{u})\int_{\gamma_k^-}e^{F+\bar{F}}w_b,$$
	where $w_a=S(A_a).$
	
	In \cite{cecotti1991topological}, using Leznov-Saveliev method, Cecotti-Vafa shows that there exist (c.f. (A.9) in \cite{cecotti1991topological}) a block-diagonal and holomorphic matrix $\mathcal{F}$ and a unit lower triangular matrix $\mathcal{N}$, such that $\T=e^{\mathcal{F}}\mathcal{N}.$ As a result, there exists a holomorphic function $\lambda(u)$, such that
	\[\int_{\gamma_k^-}e^Fdz_1\wedge\cdots\wedge dz_n=e^{\lambda(u)}\int_{\gamma_k^-}e^{F+\bar{F}}w_0.\]
	
	However, by introducing some $U(1)$ actions and using the Agmon estimate (see, for example, \cite{DY2020cohomology} or \cref{witagm}), we have the following structure theorem for the matrix-valued functions $\T$. Moreover, we show $\lambda(u)\equiv0:$
	
	Let $\mu=\dim\Jac(f)=\dim\Jac(F)(\cdot,u)$ (if $|u|$ is small). To avoid heavy notation, $\Jac(F(\cdot,u))$ is also denoted by $\Jac(F).$ 
	\begin{thm}\label{lgcy1} When we choose a specific monomial representation $\{\phi_a\}_{a=0}^{\mu-1}$ of a basis in $\Jac(F)$ and require the elements of that basis listed in a specific order (See \cref{ttgeometry}), then the matrix-valued function $\T$ is a (block) unit lower triangular matrix. More explicitly,
		\begin{align*}
		(\T)_{aa}&=1 \quad\mbox{ for $0\leq a\leq \mu-1$};\\
		(\T)_{ab}&=0 \quad\mbox{ if $a<b$};\\
		(\T)_{ab}&=\begin{cases}
		&0, \quad\mbox{ if $\frac{l(A_a)-l(A_b)}{n}\notin \mathbb{Z}^+$};\\
		&\frac{(-1)^{l_{ab}}}{l_{ab}!}\sum_{l_c=l_b}\int_{\C^n}\bar{F}^{l_{ab}}A_a\wedge*\bar{w}_c h^{cb}, \quad\mbox{ if $l_{ab}:=\frac{l(A_a)-l(A_b)}{n}\in \mathbb{Z}^+$.}
		\end{cases}
		\end{align*}
		
		Here $(h^{ab})$ is the inverse of $(h_{ab})$, $h_{ab}:=h(w_a,w_b).$
		In particular,
		\[\int_{\gamma_k^-}e^{F+\bar{F}}w_0=\int_{\gamma_k^-}e^Fdz_1\wedge\cdots\wedge dz_n.\]
	\end{thm}
	\begin{rem}\label{diff}
	    Our approach to constructing the CY/LG correspondence differs from that presented in \cite{fan2022constructing}. In their paper, they modified the real structure on the CY side and established an isomorphism map denoted by $R_{FLY}$ between the $tt^*$ structures on the LG side given by harmonic forms and the $tt^*$ structures on the CY side given by primitive forms (with modified real structure). However, in our paper, we use the transition matrix $\T$ to define the isomorphism $r$, 
     while keeping the natural real structures on both sides, i.e., the obvious complex conjugation.

Additionally, we can describe the $tt^*$ structure on $\Jac(F)'$ directly (rather than on the space of harmonic forms), with real structures modified from \cite[(4.2)]{carlson1980infinitesimal} (see \cref{last}). In this case, the natural residue map $R$ provides a complete isomorphism of $tt^*$ structures. In particular, it preserves the real structure as well.
	\end{rem}
	
	Let
	$$X_{F(\cdot,u)}=\{[z]\in \C P^{n-1}:F(z,u)=0 \}.$$
To avoid heavy notations, sometimes $X_{F(\cdot,u)}$ is also denoted by $X_F.$

 In \cref{five}, we define a map $R_u:\Jac(F)'\to H^{n-2}(X_F)_{\prim}[-1]$ that preserves the Hodge filtration. In particular, $R_u(1)$ is a nowhere vanishing holomorphic $(n-2)$-form on $X_{F}$. Now let $G^{CY}:=-\partial\bar{\p}\log\int_{X_{F(\cdot,u)}}R_u(1)\wedge*\overline{R_u(1)}$, then $G^{CY}$ is the Weil-Petersson metric on $M.$ It follows from Theorem \ref{lgcy1},  that
	\begin{thm}\label{133}
		$G=G^{CY}$(See Theorem \ref{131} for the definition of the metric $G$).
	\end{thm}
	
	Finally, we show the full CY/LG correspondence
	\begin{thm}\label{134}
		The small $tt^*$ structure on LG's side (See Definition \ref{smallLG}) and the $tt^*$ structure on CY's side (See in the last subsection) are isomorphic.
	\end{thm}
	\begin{rem}
		The existence of a map between Hodge structures on the LG and CY sides that preserves the Hodge filtration and bilinear form is widely acknowledged \cite{carlson1980infinitesimal,Fan2020LGCYCB,Hertling2004} . Therefore, the focus now turns to the preservation of the real structures. We can actually modify the real structure on $\Jac(F)$ given by Cecotti in \cite{CECOTTI1991N} using the transition matrix obtained in Theorem \ref{lgcy1}. As shown in \cref{last}, the modified real structure is preserved by the map constructed in \cite[Theorem 7.3]{Hertling2004}(which is the map $R$ constructed in Definition \ref{resmap} ).
	\end{rem}

	\subsection{Organizations}
	
	The following is the organization of the paper. In \cref{two}, we summarize the Griffiths-Carlson theory. In \cref{three}, we introduce the $tt^*$ geometric structure on the Hodge bundle in the Landau-Ginzburg B-model and use the spectral flow $S$ to establish a relationship between the space of harmonic forms and the Jacobi ring. Then,  we explore the vacuum line bundle and the Weil-Petersson metric on Landau-Ginzburg's side. Moreover, we would like to emphasize that following this section, we choose a specific monomial representation of basis of $\Jac(F)$ and required them listed in a specific order (See \cref{ttgeometry}). \cref{four} discusses the Lefschetz thimbles and investigates two different types of period integrals.  Lastly, we prove the Calabi-Yau/Landau-Ginzburg correspondence for $tt^*$ structure in \cref{five} using period integrals. In the appendix, we summarize some useful results about Agmon estimate, and prove the Riemann bilinear formula. Moreover, we modify the real structure on $\Jac(F)'$ given in \cite{CECOTTI1991N} and show that it is preserved by the residue map $R.$\\

	\noindent \textbf{Acknowledgments}.  The authors thank Huijun Fan for useful discussions on related subjects.  Part of this work was done while J.Y. was visiting Peking University in the spring of 2021. J.Y. thanks for their hospitality and provision of an excellent working environment. J. Y. also values Xianzhe Dai's encouraging and thought-provoking conversations.
	
	\section{The Griffiths-Carlson theory}\label{two}

	For simplicity, now assume that $f$ is homogeneous of degree $n,$ then $$X_f:=\{[z_1,...,z_n]\in \mathbb{P}^n:f(z_1,...,z_n)=0\}$$
	is a Calabi-Yau $(n-2)$-fold. The state space of CY B-model on $X_f$ is given by $$H^{n-2}(X_f)_{\prim}.$$ 
	
	It is known that, as a vector space, the state space of Landau-Ginzburg model is isomorphic to the Jacobi ring $\Jac(f)$ of $f$, where
	$$\Jac(f)=\C[z_1,\ldots,z_n]/\langle\p_1f,\ldots,\p_nf\rangle.$$
	
	It follows from Proposition \ref{res} that
	the (small) state spaces of LG B-model and CY B-model are related via the residue map.
	More explicitly, let $$\Omega=\sum_{i=1}^n(-1)^{i+1}z_idz_1\wedge\cdots\wedge\widehat{dz_i}\wedge\cdots\wedge dz_n,$$
	be a nonzero section of $\Omega_{\mathbb{P}^{n-1}}\cong\mathcal{O}_{\mathbb{P}^{n-1}}(-n)$,
	and assume that $\phi$ is a homogeneous polynomial of degree $kn-n$, then a rational $(n-1)$-form with polar locus $X_f$ can be constructed as:
	$$\Omega_{\phi}=\frac{\phi\Omega}{f^k}\in \Omega^{n-1}(\mathbb{P}^{n-1}-X_f).$$
	
	Here for a complex manifold $X$, $\Omega^k(X)$ denotes the space of holomorphic $k$-forms.

	There is a well known pole order filtration on the space of rational differential $(n-1)$-forms, 
	$$\F_{\ord}^1\subset\F_{\ord}^2\subset\cdots\subset\F_{\ord}^{n-1},$$
	where $\F_{\ord}^k$ is the subgroup with a pole of order $\leq k$ along $X_f$.
	
	The residue map is defined to be
	\begin{align*}
	\res:H^{n-1}(\mathbb{P}^{n-1}-X_f)&\rightarrow H^{n-2}(X_f,\C)\\
	\frac{1}{2\pi i}\int_{\tau(\gamma)}\Omega_{\phi}&=\int_{\gamma}\res \Omega_{\phi}
	\end{align*}
	where $\tau:H_{n-2}(X_f)\rightarrow H_{n-1}(\mathbb{P}^{n-1}-X_f)$ is the Leray coboundary map, i.e. $\tau(\gamma)$ is a tube neighborhood of $\gamma$.

		The Griffiths-Carlson theory 
	\cite{carlson1980infinitesimal} provides the CY/LG correspondence for the cup product and the pairing:
	\begin{prop}[\cite{carlson1980infinitesimal}]\label{res}\label{carlson} The residue map satisfies the following properties:
		\begin{enumerate}
			\item The image of the residue map is the $(n-2)$-th primitive cohomology of $X_f$;
			\item The residue map takes the pole order filtration on the space of rational differentials to the Hodge filtration on $H^{n-2}(X_f;\C)_{\prim}$, i.e., if $\phi$ is a homogeneous polynomial of degree $n(k-1),$
			$$\res\Omega_{\phi}\in\F^{n-k-1}H^{n-2}(X_f;\C)_{\prim}.$$
			Moreover, the residue of a form $\Omega_{\phi}$ has Hodge level $n-k$  if and only if $\phi$ lies in the Jacobian ideal $\langle\p_1f,\cdots,\p_nf\rangle$.
		\end{enumerate}
	Let $\phi$ and ${\psi}$ be homogeneous polynomials of degree $n(a-1)$ and $n(b-1)$ respectively, then
	\begin{enumerate}
		\item $\int_{X_f}\res\Omega_{\phi}\wedge\res\Omega_{\psi}=0$ if and only if $\phi\psi$ lies in the Jacobian ideal $\langle\p_1f,\cdots,\p_nf\rangle$;
		\item Let $\delta:H^{n-2}(X_f)\rightarrow H^{n-1}(\mathbb{P}^{n-1})$ be the coboundary operator in the Poincar\'{e} residue sequence, then the cup product in the projective space is given by
		$$\delta(\res\Omega_{\phi}\wedge\res\Omega_{\psi})=c_{ab}\frac{\phi\psi\Omega}{\p_1f\cdots\p_nf},$$
		where $c_{ab}=\frac{(-1)^{a(a-1)/2+b(b-1)/2+n+(b-1)^2}}{(a-1)!(b-1)!}\deg f$.
		\item 
		$$\int_{X_f}\res\Omega_{\phi}\wedge\res\Omega_{\psi}=\frac{\tilde{c}_{ab}}{(2\pi i)^n}\int_{\Gamma(\epsilon)}\frac{\phi\psi dz_1\wedge\cdots\wedge dz_n}{\p_1f\cdots\p_nf},$$
		where $\tilde{c}_{ab}=\frac{(-1)^{a(a-1)/2+b(b-1)/2+n+(b-1)^2}}{(a-1)!(b-1)!},$ $\Gamma(\epsilon)=\{|\p_if|=\epsilon: i=1,\ldots,n\}$, and the right hand side is the Grothendieck residue symbol.
	\end{enumerate}
	\end{prop}

	\section{Deformation Theory and Weil-Petersson Metric on the Marginal Deformation Space}\label{three}
	
	\subsection{Hodge decomposition}\label{Hodgedecomposition}
	
	\begin{defn}[Quasi-homogeneous polynomials]\label{quasi}
		We say a polynomial $f: \C^n\rightarrow\C$ is quasi-homogenous (of $(q_1,...,q_n)$ for $q_1,\ldots,q_n\in\mathbb{Q}$), if for any $\lambda\in\C^*$,
		$$f(\lambda^{q_1}z_1,\ldots, \lambda^{q_n}z_n)=\lambda f(z_1,...,z_n).$$
		Each $q_i$ is called the weight of $z_i$, and let $\w(z_i):=q_i$. We extend the definition of $\w$ to the monomial $\prod_{i=1}^nz_i^{k_i}$,  such that $\w(\prod_{i=1}^nz_i^{k_i})=\sum_{i=1}^nk_iq_i$. In particular,
		$\w(f)=1$.
		
		In the physical literature, the central charge $\hat{c}_f$ is defined to be
		$$\hat{c}_f=\sum_{i=1}^n(1-2q_i).$$
	\end{defn}
	

	While, for convenience, we assume that $f$ is non-degenerate, i.e. we require that
	\begin{enumerate}
		\item $f$ contains no monomial of the form $z_iz_j$ for $i\neq j$,
		\item $f$ has only an isolated singularity at the origin.
	\end{enumerate}
	
	Next, we'll go over some physics terminology. We consider the pair $(\C^n,f)$. The polynomial $f$ is called the superpotential for LG model. The Schr\"{o}dinger representation of 1d LG model consists of a Hilbert space, which is the space of $L^2$-integrable forms $L^2\Lambda^*(\C^n)$, and four charge operators $Q_{+},Q_{-}, Q_{+}^{\dag},Q_{-}^{\dag}$, which are operators on $L^2\Lambda^*(\C^n)$ given by
	\begin{flalign*}
	&Q_{+}=\bar{\p}_{f}=\bar{\p}+df\wedge,\quad Q_{-}=\p_{f}=\p+d\bar{f}\wedge,\\
	&Q_{+}^{\dag}=\bar{\p}_{f}^{\dag}=-\ast \p_{-f}\ast, \quad Q_{-}^{\dag}=\p_{f}^{\dag}=-\ast \bar{\p}_{-f}\ast.
	\end{flalign*}
	Here $\ast$ is the Hodge star operator with respect to the standard hermitian metric on $\C^n$. Then we have the following relation, called supersymmetric algebra structure in the physical literature
	\begin{flalign*}
	&\p_{f}^2=\p_{f}^{\dag 2}=\overline{\p}_{f}^2=\overline{\p}_{f}^{\dag 2}=0,\\
	&\{\p_{f},\p_{f}^{\dag}\}=\{\overline{\p}_{f},\overline{\p}_{f}^{\dag}\}=\Delta_{f},\\
	&\{\overline{\p}_{f},\p_{f}\}=\{\overline{\p}_{f}^{\dag},\p_{f}^{\dag}\}=\{\p_{f},\overline{\p}_{f}^{\dag}\}=\{\overline{\p}_{f},\p_{f}^{\dag}\}=0,
	\end{flalign*}
	where $\{\cdot,\cdot\}$ is the anti-commutative bracket, that is, $\{P,Q\}=PQ+QP$.
	
	In the spirit of Cecotti-Vafa's paper \cite{cecotti1991topological}, Klimek-Lesniewski \cite{Klimek1991local} and Fan \cite{fan2011schr} used the $L^2$ analysis to study the $L^2$-cohomology $H^{*}_{(2),\overline{\p}_{f}}(\C^n)$ of $(L^2\Lambda^*(\C^n), \overline{\p}_{f})$ and the harmonic forms $H_f^*:=\ker\Delta_{f}$. Let us recall their results as follows.
	
	In \cite{Klimek1991local}, Klimek-Lesniewski studied the system $(\C^n, V=\sum_{j=1}^nV_j(z)dz_j)$, where $V$ is a holomorphic 1-form satisfying the elliptic condition
	\begin{enumerate}
		\item $|V(z)|^2:=\sum_{j=1}^n|V_j(z)|^2\rightarrow\infty$, as $|z|\rightarrow\infty$;
		\item For any $\epsilon>0$, there is a constant $C$ such that
		$$|\p_jV_k(z)|\leq\epsilon|V(z)|^2+C,\quad\text{for all }z\in\C^n,\text{ and }j,k=1,\ldots,n.$$
	\end{enumerate}
	In particular, they pointed out that when $V=d f$, where $f:\C^n\rightarrow \C$ is a holomorphic polynomial, their Laplacian coincides with $\Delta_f$ above.
	
	In \cite{fan2011schr}, Fan studied a more general system $((M,g),f)$, where 
	\begin{enumerate}
		\item $(M,g)$ is a K\"ahler manifold with bounded geometry;
		\item $f$ is a holomoprhic function on $M$ satisfying the following strongly tame condition: for any constant $C>0$, there is
		$$|\nabla f|^2(z)-C|\nabla^2f|(z)\rightarrow\infty,\text{ as } d(z,z_0)\rightarrow\infty$$ \text{ for a fix point $z_0\in M$}
	\end{enumerate}
	
	In particular, Fan proved that the following two cases are such strongly tame systems
	\begin{itemize}
		\item $(\C^n,f)$, where $f$ is a non-degenerate quasi-homogeneous polynomial on $\C^n$;
		\item $((\C^*)^n,f)$, where $f$ is a convenient and non-degenerate Laurent polynomial defined on the algebraic torus $(\C^*)^n$. \footnote{Sabbah studied the corresponding deformation theory in an algebraic formulation, see for example \cite{Sabbah2008}.}
	\end{itemize}
	
	According to Fan's proof, one can also check that $df$ satisfy the above elliptic condition when $f$ is a non-degenerate quasi-homogeneous polynomial. Thus, one has the Hodge theorem in such a special case.
	

	\begin{thm}[Hodge Theorem \cite{fan2011schr}, \cite{Klimek1991local}]\label{Hodgethm} Let $f\in\C[z_1,\ldots,z_n]$ be a non-degenerate quasi-homogeneous polynomial. Then
		\begin{enumerate}
			\item $\dim H_f^*\leq\infty$.
			\item \emph{(Hodge decomposition)} There are orthogonal decompositions
			$$L^{2}\Lambda^{*}(\C^n)=H_{f}^{*}\oplus\Im(\overline{\p}_{f})\oplus \Im(\overline{\p}_{f}^{\dag}).$$
			More precisely, there is a self-adjoint compact operator $G_{f}$ on $L^2\Lambda^*(X)$ such that
			$$L^2\Lambda^k(\C^n)=H_f^k\oplus \overline{\p}_{f}\left(\overline{\p}_{f}^{\dag}G_{f}L^2\Lambda^k(\C^n)\right)\oplus \overline{\p}_{f}^{\dag}\left(\overline{\p}_{f}G_{f}L^2\Lambda^k(\C^n)\right).$$
			\item \emph{(Vanishing Theorem)}
			\begin{equation}
			\label{eq:abs}
			H_{f}^{k}\cong H^k_{(2),\overline{\p}_{f}}(\C^n)\cong
			\begin{cases}
			0& \text{if $k \neq n$}\\
			\Omega^{n}(\C^n)/df\wedge\Omega^{n-1}(\C^n)\cong\Jac(f) & \text{if $k=n$}.
			\end{cases}
			\end{equation}
		\end{enumerate}
	\end{thm}

	\begin{rem}\label{KHD} One can also introduce the Witten deformation $d_{2\Re(f)}:=d+2d\Re(f)\wedge$ of de Rham differential $d$. By K\"ahler identity,
		$$\Delta_{2\Re(f)}=\{d_{2\Re(f)},d_{2\Re(f)}^{\dagger}\}=2\Delta_f.$$
		Moreover, we have the Kodaira-Hodge decomposition, see also for example, \cite{DY2020cohomology},
		$$L^{2}\Lambda^{*}(\C^n)=\ker(\Delta_{2\Re(f)})\oplus \Im(d_{2\Re(f)})\oplus \Im(d_{2\Re(f)}^{\dag}).$$
		
		In addition, Fan \cite{fan2011schr} shows that $*\Delta_f=\Delta_{-f}*$ (where $*$ is the Hodge star operator), which plays a crucial role in the study of torsion type invariants in the LG model \cite{fanfang2016torsion}.
	\end{rem}
	

	\subsection{$tt^*$ geometry}\label{ttgeometry}
	
	Fix  a monomial representation $\{\phi_a\}_{a=0}^{\mu-1}$ of basis of $\Jac(f)$, where $\mu=\dim\Jac(f)$ and $\phi_0=1$. Moreover, we assume that $\w(\phi_a)\leq \w(\phi_b)$ if $a<b$ (See Definition \ref{quasi} for the definition of $\w$). In particular, let $\{\psi_i:=\phi_{a_i}\}_{i=1}^s$ be the subset of $\{\phi_a\}_{a=0}^{\mu-1}$, such that each $\psi_i$ has weight $1$, $i=1,\ldots,s$.  We will study marginal deformation of $f$
	$$F(z,u)=f(z)+\sum_{i=1}^{s}u_{i}\psi_i(z).$$
	
	We denote  the space of parameters ${u_{i}}$ by $M$, which could be taken as a small neighborhood of the origin in $\C^s$. Therefore, we have a family of supersymmetric algebra operators $\bar{\p}_{F},\p_{F}, \bar{\p}^{\dag}_{F},\p_{F}^{\dag},\Delta_{F}$ parameterized by $u\in M$. When $|u|$ is small enough, the results in the previous subsection apply to $\Delta_F$, see a rigourous proof in \cite{fan2011schr}. In particular, when $|u|$ is small, $\Jac(F(\cdot, u))$ is isomorphic to $\Jac(f)$ as a vector space, and $\{\phi_a\}_{a=0}^{\mu-1}$ is still a representation of basis of $\Jac(F(\cdot,u)).$ To avoid heavy notation, we will denote $\Jac(F(\cdot,u))$ as $\Jac(F).$

	We have the trivial complex Hilbert bundle $L^2\Lambda^{*}(\C^n)\times M\rightarrow  M$. For simplicity, denote by $L^2\mathcal{A}$ its $L^2$-integrable section space. There are two natural parings,
	\begin{flalign*}
	h: ~&L^2\A\times L^2\mathcal{A}\longrightarrow C^{\infty}(M)\qquad h(\alpha,\beta)=\int_{X}\alpha\wedge\ast\bar{\beta},\\
	\eta: ~&L^2\mathcal{A}\times L^2\mathcal{A}\longrightarrow C^{\infty}(M)\qquad \eta(\alpha,\beta)=\int_{X}\alpha\wedge\ast\beta.
	\end{flalign*}
	
	Note that there is a canonical real structure on the Hilbert bundle induced from the complex conjugate in $L^2\Lambda^*(\C^n)$, which was denoted by $\bar{\cdot}$. Then $h(\alpha,\beta)=\eta(\alpha,\bar\beta)$.
	
	Let ${H}^n:=H_{F}^{n}$ be the Hodge bundle over $M$. Fiberwisely,  $H^n\big|_u$ is the space of harmonic forms of $\Delta_{F(\cdot,u)}$. We denote the space of its section by $\H$.
	
	Let $\Pi_u: L^2\A\rightarrow \H$ be the harmonic projection, $G$ be the inverse of $\Delta_{F}$ on $\Im(\bar{\p}_{F})\oplus\Im(\bar{\p}_{F}^{\dag})$, then $G$ commutes with the operators $\bar{\p}_{F},\p_{F}, \bar{\p}^{\dag}_{F},\p_{F}^{\dag},\Delta_{F}$, and the Hodge decomposition reads
	$$\Id=\Pi_u+\Delta_{F}G=\Pi_u+G\Delta_{F}.$$

	Now we introduce some important operators on the Hodge bundle that give the $tt^*$ structures. Let $u=(u_1,...,u_s)$ be local coordinates of  $M$, and $\partial_i:=\partial_{u_i}$. 
	
	\vskip 0.1cm
	\noindent(1) The connection $D$, $\bar{D}$
	\vskip 0.1cm
	Notice that the Hodge bundle is embedded into the trivial Hilbert bundle, so we can define $D$, $\bar{D}$ in a natural way:
	$$D_i:=D_{\frac{\p}{\p u_i}}=\Pi_u\circ \p_i,\quad \bar{D}_{\bar{i}}:=\bar{D}_{\frac{\p}{\p\bar{u_i}}} =\Pi_u\circ\bar{\p}_{\bar{i}}\quad i=1,...,s.$$

	\noindent(2) The operators $C_i, \bar{C}_{\bar{i}}$
	\vskip 0.1cm
	We define $C_{i}=\Pi_u\circ\p_{i}F=\Pi_u\circ\psi_i$, $\bar{C}_{\bar{i}}=\Pi_u\circ\overline{\p_iF}=\Pi_u\circ\overline{\psi_i}$. We can easily show that
	$$C_{i}\alpha_{a}=(\p_{i}F\alpha_{a})-\bar{\p}_{F}\bar{\p}^{\dag}_{F}G(\p_{i}F)\alpha_{a}, \quad \bar{C}_{\bar{i}}\alpha_{a}=(\overline{\p_{i}F}\alpha_{a})-\p_{ F}\p^{\dag}_{F}G(\overline{\p_{i}F})\alpha_{a}.$$
	By definition, $\bar{C}_{\bar{i}}$ is the adjoint operator of $C_i$ with respect to the $tt^*$ metric $h$, i.e.
	$$h(C_i\alpha_a,\alpha_b)=h(\alpha_a,\bar{C}_{\bar{i}}\alpha_b).$$
	\begin{prop}[$tt^*$ equation]  $D_i,\bar{D}_{\bar{j}}$, $C_i, \bar{C}_{\bar{j}}$ satisfy the following equations
		\begin{enumerate}
			\item $[C_{i},C_{j}]=0, \quad [\bar{C}_{\bar{i}}, \bar{C}_{\bar{j}}]=0,\quad [D_{i},\bar{C}_{\bar{j}}]=[\bar{D}_{\bar{i}},C_{j}]=0;$
			\item $[D_{i},C_{j}]=[D_{j},C_{i}], \quad [\bar{D}_{\bar{i}},\bar{C}_{\bar{j}}]=[\bar{D}_{\bar{j}},\bar{C}_{\bar{i}}]$;
			\item $[D_i, D_j]=0,\quad [\bar{D}_{\bar{i}},\bar{D}_{\bar{j}}]=0,\quad [D_{i},\bar{D}_{\bar{j}}]=-[C_{i},\bar{C}_{\bar{j}}].$
		\end{enumerate}
	\end{prop}
	
	\begin{proof} See, for example, \cite{Hao} or \cite{Tang2017ttGS}.
	\end{proof}
	
	\subsection{The isomorphism $S$}\label{S}

	
	Now take a $C^{\infty}(M)$-basis $\{A_{a}(t)\}_{a=1}^{\mu}$ of $C^{\infty}(M)\otimes\Omega^{n}_{\C^n\times M/M}/dF\wedge\Omega^{n-1}_{\C^n\times M/M}$\footnote{where $\Omega^*_{\C^n\times M/M}$ is the sheaf of relative holomorphic differential forms.}, we will construct a basis of $\H$ from it. To do this, we follow the method in \cite{li2013primitive} to show that
	$$H^{n}_{\bar{\p}_{F}}(\C^n)\cong H^{n}_{c,\bar{\p}_{F}}(\C^n),$$
	where the left side is the cohomology of smooth complex, and the right side is the cohomology of smooth complex having a compact support.
	
	Let $(\C^n)^*:=\C^n-\{0\}$, we define the operator$$V_{F}=\sum_{\nu=1}^{n}\frac{\overline{\p_{\nu}F}}{|\nabla F|^2}(dz_{\nu})^{\dagger}=\frac{(\p F)^{\dagger}}{|\nabla F|^2}:\quad \mathcal{A}^{*,*}((\C^n)^*)\rightarrow \mathcal{A}^{*-1,*}((\C^n)^*),$$
	where for a complex manifold $X$, $\mathcal{A}^{p,q}(X)$ denotes the space of smooth $p,q$ forms on $X.$

	Next, given a cut-off function $\rho$, which is compactly supported in a small disc around the origin, we define
	$$T_{\rho}: \mathcal{A}^{*}(\C^n)\rightarrow \mathcal{A}_{c}^{*}(\C^n),\qquad R_{\rho}:\mathcal{A}^{*}(\C^n)\rightarrow \mathcal{A}^{*}(\C^n)$$
	by
	$$T_{\rho}(A)=\rho A+(\bar{\p}\rho)V_{F}\frac{1}{1+[\partial, V_{F}]}(A),\qquad R_{\rho}(A)=(1-\rho)V_{F}\frac{1}{1+[\partial, V_{F}]}(A),$$
	where for a manifold $X$, $\mathcal{A}^*(X)$ and ($\mathcal{A}_c^*(X)$) denotes the space of smooth $p$ forms (with compact support) on $X.$ 
	
	\begin{lem}[\cite{li2013primitive}]
		$[\bar{\p}_{F}, R_{\rho}]=1-T_{\rho}$ on $\mathcal{A}^*(\C^n)$. Moreover, the embedding $(\mathcal{A}_{c}^*(\C^n),\bar{\p}_{F})\hookrightarrow (\mathcal{A}^*(\C^n),\bar{\p}_{F})$ is a quasi-isomorphism.
	\end{lem}

	Let $A$ be a representative of $C^{\infty}(M)\otimes\Omega^n_{\C^n\times M/M}/dF\wedge\Omega^{n-1}_{\C^n\times M/M}$, then $T_{\rho}(A)$ is $L^2$-integrable, thus we can define $w$ to be its harmonic projection, that is,
	$$w=\Pi_u(T_{\rho}(A)).$$
	
	For simplicity, we write it as $w=A+\bar{\p}_{F}\nu$, where, by construction, $\nu$ is a $n-1$ form with polynomial growth. We define the following map
	\begin{align*}
	S:C^{\infty}(M)\otimes\Omega^n_{\C^n\times M/M}/dF\wedge\Omega^{n-1}_{\C^n\times M/M}&\longrightarrow \mathcal{H},\\
	A&\longmapsto w.
	\end{align*}
	
	\begin{prop}\label{spectral} The map $S$ satisfies the following properties:
		\begin{enumerate}
			\item $S$ is a $C^{\infty}(M)$-linear map;
			\item $S$ is well defined, that is, $S$ is independent of the choice of the representative in $C^{\infty}(M)\otimes\Omega^n_{X\times M/M}/dF\wedge\Omega^{n-1}_{X\times M/M}$ and also independent of the cut-off function $\rho$.
		\end{enumerate}
	\end{prop}
	
	\begin{proof} \ 
		\begin{enumerate}[(1)]
			\item The $C^{\infty}(M)$-linearity follows from the definitions of $T_{\rho}$ and $\Pi_u$.
			
			\item 
			\begin{itemize}
				\item The independence of the choice of the representative:
				
				Let $A'=A+dF\wedge B$, for some $B\in C^{\infty}(M)\otimes\Omega_{\C^n\times M/M}^{n-1}$ with polynomial growth. By definition, there exist $n-1$ forms $\nu$ and $\nu'$ with polynomial growth such that
				$$w:=S(A)=A+\bar{\p}_F\nu,\quad w':=S(A')=A'+\bar{\p}_F\nu'.$$
				Moreover, since $B$ is holomorphic in the $\C^n$ direction,
				$$w':=S(A')=A+dF\wedge B+\bar{\p}_F\nu'=A+\bar{\p}_F(B+\nu').$$
				Then one has
				\[w-w'=\bar{\p}_{F(\cdot,u)}(\nu+B-\nu'),\]
				where $(\nu+B-\nu')$ has polynomial growth. Then by Agmon estimate (Lemma \ref{epde12}) and integration by parts, one has
				\[h(w-w',w-w')=h(w-w',\bar{\p}_{F(\cdot,u)}(\nu-\nu'))=h(\bar{\p}_{F(\cdot,u)}^{\dagger}(w-w'),\nu-\nu')=0,\]
				which implies that $w=w'$.\\
				\item The independence of the choice of the cut-off function:
				
				Let $S$ and $S'$ be the map with respect to the cut-off functions $\rho$ and $\rho'$ respectively. Then
				$$A=S(A)-\bar{\p}_{F(\cdot,u)}\nu=S'(A)-\bar{\p}_{F(\cdot,u)}\nu'$$ for some differential forms $\nu$ and $\nu'$ with polynomial growth.
				Using the same argument as above, one can easily verify that $S(A)=S'(A)$.
				
			\end{itemize}
		\end{enumerate}
	\end{proof}
	
	\begin{rem} In the physical literature, such as in \cite{CECOTTI1991N}, $S$ is called the spectral flow.
	\end{rem}
	
	Let $u=(u_1,...,u_s)$ be local coordinates of  $M$, and $\partial_i:=\partial_{u_i}$.
	If $A$ is holomorphic in $u\in M$, then $w:=S(A)=A+\bar{\p}_F\nu$ is holomorphic. Since 
	$$\bap_{\bar{i}}w=\bap_{\bar{i}}A+\bap_{\bar{i}}\bar{\p}_{F}\nu=\bar{\p}_F\bap_{\bar{i}}\nu\in\Im(\bar{\p}_{F})\quad (\mbox{Note that $[\bap_{\bar{i}}, \bar{\p}_{F}]=0$}).$$

	In particular, when $|u|$ is small, we can simply choose
	$$A_a=\phi_a dz_1\wedge\cdots\wedge dz_n, \quad \{\phi_a\}_{a=0,\ldots,\mu-1}\text{ is a monomial representation of basis of}\Jac(f).$$
	Moreover, we let $\phi_0=1$. Since $|u|$ is small, $\{\phi_a\}_{a=0,\ldots,\mu-1}$ is also a monomial representation of basis of $\Jac(F)$. Then $A_a$'s are $\bar{\p}_F$-closed and one can find a unique harmonic form $w_a=S(A_a)$ such that
	\begin{equation}\label{dec}
	\phi_a dz_1\wedge\ldots\wedge dz_n=w_a-\bar{\p}_F\nu_a
	\end{equation}
	for some $n-1$ form $\nu_a$, $a=0,\ldots,\mu-1$. Here $\nu_a$ has at most polynomial growth.
	
	Thus, we get a holomorphic basis $\{w_a\}$ for $\H$.

	Now we would like to explore the relationship between the operator $C$ and the product in the Jacobi ring of $F$:
	
	On the one hand, by the definition, $\partial_{i}Fw_{a}=C_{ia}^{b}w_{b}+\bar{\partial}_{F}\nu$ for some differential forms $\nu$ with polynomial growth and matrix $C_{ia}^b$. On the other hand, there exists $\tilde{C_{ia}^b}\in \mathcal{O}_M$, such that \be\label{productjac}\partial_{i}F\phi_{a}=\widetilde{C}_{ia}^{b}\phi_{b}\ee
	in $\Jac(F)$.
	
	\begin{align*}
	\partial_{i}Fw_{a}&=\partial_{i}F(A_{a}+\bar{\partial}_{F}\nu_a)\\
	&=\partial_{i}FA_{a}+\bar{\partial}_{F}\partial_iF\nu_a\\
	&\overset{(a)}=(\widetilde{C}_{ia}^{b}A_{b}+dF\wedge\delta_{a})+\bar{\partial}_{F}\nu_a\\
	&\overset{(b)}=\widetilde{C}_{ia}^{b}\left(w_{b}-\bar{\partial}_{F}\nu_b+\bar{\partial}_{F}\delta_a\right)+\bar{\partial}_{F}\nu_a=\widetilde{C}_{ia}^{b}w_{b}+\bar{\partial}_{F}\widetilde{\nu},
	\end{align*}
	Here for the equality (a), it follows from (\ref{productjac}) that $\partial_{i}FA_{a}=\widetilde{C}_{ia}^{b}A_{b}+dF\wedge\delta_{a}$ for some $\delta_{a}\in\Omega_{X\times M/M}^{n-1}$. For the equality (b), notice that $dF\wedge\delta_{a}=\bar{\partial}_{F}\delta_{a}$.
	By the similar arguments as in Proposition \ref{spectral}, we get $C_{ia}^{b}=\widetilde{C}_{ia}^{b}$, i.e., the operator $C$ coincides with the product in the Jacobi ring $\Jac(F)$.


	\subsection{$U(1)$ action}\label{2.4}
	To further study the structure of the Hodge bundle, let us introduce the $U(1)$ action associated to $f$ on differential forms, which is closely related to the pole order filtration in \cref{two}.
	
	First, assume that the weight $q_i=\frac{a_i}{b_i}$ with $(a_i,b_i)=1$, and $d=\mathrm{lcm}\left(b_{1}, \ldots, b_{n}\right)$, i.e. $d$ is the least common multiple of $b_1,...,b_n$. We put $Q_{i}=q_{i} d$.
	
	Let us consider the $U(1)$ action on $\mathbb{C}^n$ which is related to the weight of the variables:
	$$e^{i\theta}\cdot(z_1,\ldots,z_n)=\left(e^{iQ_1\theta}z_1,\ldots,e^{iQ_n\theta}z_n\right).$$
	In particular, if $f$ is homogeneous,
	$$e^{i\theta }\cdot(z_1,\ldots,z_n)=\left(e^{i\theta }z_1,\ldots,e^{i\theta }z_n\right).$$

	This action induces a unitary action $\T$ of $U(1)$ on differential forms, such that for any $(p,q)$-form $\alpha=\alpha(z_1,...,z_n,\bar{z}_1,...\bar{z}_n)dz_{i_1}...dz_{i_p}d\bar{z}_{j_1}...d\bar{z}_{j_q}$, \begin{align*}&\ \ \ \ \T(\theta)\alpha\\&=e^{i\theta\left(-\sum_{k=1}^pQ_{i_k}+\sum_{k=1}^q{Q_{j_k}}\right)}\alpha(e^{-i\theta Q_1}z_1,...,e^{-i\theta Q_n}z_n,e^{i\theta Q_1}\bar{z}_1,...,e^{i\theta Q_n}\bar{z}_n)dz_{i_1}...dz_{i_p}d\bar{z}_{j_1}...d\bar{z}_{j_q}.\end{align*}
	Note that $f$ is a $(0,0)$-form, thus,
	$$\T(\theta)f=e^{-id\theta}f.$$
	Similarly, if $f$ is homogeneous,
	$$\T(\theta)\alpha=e^{i(q-p)\theta}\alpha(e^{-i\theta}z_1,...,e^{-i\theta}z_n,e^{i\theta}\bar{z}_1,...,e^{i\theta}\bar{z}_n)dz_{i_1}...dz_{i_p}d\bar{z}_{j_1}...d\bar{z}_{j_q}.$$
	
	Now consider the following shifted $U(1)$ action $\P$ on $L^2\Lambda^{p,q}(\C^n)$, $$\P(\theta)(\alpha)=e^{idp\theta}\T(\theta)\alpha,\quad \alpha\in L^2\Lambda^{p,q}(\C^n).$$
	Similarly, we have another shifted $U(1)$ action $\Q$ on $L^2\Lambda^{p,q}(\C^n)$,
	$$\Q(\theta):=e^{idq\theta}\T(-\theta).$$
	
	Immediately, we have,
	
	\begin{lem}\label{Atheta} The $U(1)$ actions $\P$ and $\Q$ satisfy the following properties:
		\begin{enumerate}[(1)]
			\item $\P$ and $\Q$ are unitary, i.e.,
			$$h(\Q(\theta)\alpha,\Q(\theta)\beta)=h(\P(\theta)\alpha,\P(\theta)\beta)=h(\alpha,\beta), \quad \alpha,\beta\in L^2\Lambda^{p,q}(\C^n).$$
			\item \label{sf6}$[\P(\theta),\bar{\partial}_F]=0,~[\P(\theta),\bar{\partial}_F^{\dagger}]=0,~[\P(\theta),\Delta_F]=0$.
			\item $\P(\theta){\partial}_F\P(\theta)^{-1}=e^{id\theta}{\partial}_F,$ $\P(\theta){\partial}_F^{\dagger}\P(\theta)^{-1}=e^{-id\theta}{\partial}_F^{\dagger}.$
			\item\label{sf7} $[\Q(\theta),{\partial}_F]=0,~[\Q(\theta),{\partial}_F^{\dagger}]=0,~[\Q(\theta),\Delta_F]=0$.
			\item $\Q(\theta)\bar{\partial}_F\Q(\theta)^{-1}=e^{id\theta}\bar{\partial}_F,$ $\P(\theta)\bar{\partial}_F^{\dag}\P(\theta)^{-1}=e^{-id\theta}\bar{\partial}_F^{\dagger}.$
		\end{enumerate}
		In particular, (\ref{sf6}) and (\ref{sf7}) imply that $\P:\H\to\H$, $\Q:\H\to\H.$
		
	\end{lem}

	\begin{rem} Let $X$ be a compact Calabi-Yau manifold, one can define the $U(1)$ actions $\P^{CY}$ and $\Q^{CY}$ as follows, for $\alpha\in \A^{p,q}(X)$,
		$$\P^{CY}(\theta)\alpha:=e^{ip\theta}\alpha,\quad \Q^{CY}(\theta)\alpha:=e^{iq\theta}\alpha.$$
		It is easy to see that $\P^{CY}$ and $\Q^{CY}$ are unitary; $\P^{CY}$ commutes with $\bar{\partial}$; $$\P^{CY}(\theta) \p\P^{CY}(\theta)^{-1}=e^{i\theta}\p.$$
		Thus one can regard $\P$ and $\Q$ as the bi-grading in the Landau-Ginzburg models. More explicitly, the infinitesimal action of $\P$ and $\Q$ plays the role of holomorphic and anti-holomorphic degrees in Landau-Ginzburg models.
	\end{rem}

	\begin{defn} For $A=z_1^{\beta_1}\cdots z_n^{\beta_n}dz_1\wedge\cdots\wedge dz_n$, we define
		$$l(A):=\sum_{i=1}^n(\beta_i+1)Q_i.$$
		We also call $l(A)$ the $U(1)$ charge of $A$.
	\end{defn}
	For the homogeneous basis $\{A_a\}_{a=0}^{\mu-1}$, set $l_a=l(A_a).$

	\begin{lem}\label{uone} The $U(1)$ actions $\P$ and $\Q$ commute with the map $S$, that is
		$$S(\P(\theta)A)=\P(\theta)(S(A)),\quad A\in C^{\infty}(M)\otimes\Omega^n_{\C^n\times M/M}/dF\wedge\Omega^{n-1}_{\C^n\times M/M}$$
		and
		$$S(\Q(\theta)A)=\Q(\theta)(S(A)),\quad A\in C^{\infty}(M)\otimes\Omega^n_{\C^n\times M/M}/dF\wedge\Omega^{n-1}_{\C^n\times M/M}.$$
		More explicitly, let $w_a=S(A_a)$, then
		\begin{align*}
		\P(\theta)(w_a)&=e^{i(nd-l_a)\theta}w_a,\\
		\Q(\theta)(w_a)&=e^{il_a\theta}w_a.
		\end{align*}
		Thus, we could define the $U(1)$ charge $l(w_a)$ of $w_a$ to be $l_a$. This lemma could be  rephrased as follows: if a harmonic form $w$ has $U(1)$ charges $l$, that is, $w=S(A)$ for some homogeneous $n$ form $A$ with $l(A)=l$, then 
		\begin{align*}
		\P(\theta)(w)&=e^{i(nd-l)\theta}w,\\
		\Q(\theta)(w)&=e^{il\theta}w.
		\end{align*}
		
	\end{lem}
	
	\begin{proof} It suffices to check the result for $A_a=\phi_adz_1\wedge\cdots\wedge dz_n$. On the one hand,
		\begin{equation}\label{sf4}\P(\theta)A_a=e^{i(nd-l_a)\theta}A_a.\end{equation}
		On the other hand, since
		\begin{equation} \label{sf1}
		A_a=S(A_a)-\bar{\p}_{F}\nu_a,
		\end{equation}
		by Lemma \ref{Atheta}, we have
		\begin{equation}\label{sf2}
		\P(\theta)A_a=\P(\theta)S(A_a)-\bar{\p}_F(\P(\theta)\nu_a);
		\end{equation}
		
		By \eqref{sf4}, the left hand side of \eqref{sf2} is
		\begin{equation}\label{sf3}
		\P(\theta)A_a=e^{i(nd-l_a)\theta}A_a=e^{i(nd-l_a)\theta}S(A_a)-\bar{\p}_F\left(e^{i(nd-l_a)\theta}\nu_a\right).
		\end{equation}

		Then (\ref{sf2}) and (\ref{sf3}) imply that
		$$e^{i(nd-l_a)\theta}S(A_a)-\P(\theta)S(A_a)=\bar{\p}_F\mu_a$$
		for some at most polynomial growth differential form $\mu_a$.
		
		By Agmon estimate (Lemma \ref{epde12}) and the fact that $[\P(\theta),\Delta_F]=0$ , we get
		$$\P(\theta)S(A_a)=e^{i(nd-l_a)\theta}S(A_a)=S(\P(\theta)A_a).$$
	\end{proof}
	
	\def\lan{\langle}
	\def\ran{\rangle}
	\def\dz{dz_1...dz_n}

	In CY's case, if $u\in \A^{p,n-p}(Y),v\in \A^{p',n-p'}(Y)$, then $\int_Yu\wedge*\bar{v}\neq0$ if and only if $p=p'$. Similarly,
	by unitary property of $\P$ and Lemma \ref{uone}, it is easy to obtain
	\begin{cor}\label{wawb} If $h(w_a, w_b)\neq0$, then $l_a=l_b$.
	\end{cor}
	\begin{proof}
		By Lemma \ref{uone},
		\[h(w_a,w_b)=h(\P(\theta) w_a,\P(\theta) w_b)=h(e^{i(nd-l_a)\theta}w_a,e^{i(nd-l_b)\theta}w_b)=e^{i(l_b-l_a)\theta}h(w_a,w_b).\]
		Hence, $h(w_a,w_b)\neq 0$ implies $l_a=l_b.$
	\end{proof}

	Now we investigate the relation between $U(1)$ charges and the real structures in the Landau-Ginzburg B-model. Since $w_a$ is a basis of $\mathcal{H}^n$, there exists a matrix-valued function $(M_{\bar{a}b}(u,\bar{u}))$ such that
	$$\bar{w}_a=\sum_bM_{\bar{a}b}(u,\bar{u})w_b,$$
	then we can define a complex anti-linear map
	$$\kappa:C^{\infty}(M)\otimes\Omega^n_{\C^n\times M/M}/dF\wedge\Omega^{n-1}_{\C^n\times M/M}\longrightarrow C^{\infty}(M)\otimes\Omega^n_{\C^n\times M/M}/dF\wedge\Omega^{n-1}_{\C^n\times M/M}$$ via the basis
	\begin{equation}\label{realform}
	\kappa(A_a):=\sum_bM_{\bar{a}b}(u,\bar{u}) A_b.
	\end{equation}
	Since $S$ is $C^{\infty}(M)$-linear, we have
	\begin{equation}\label{taukappa}
	S(\kappa(A_a))=\sum_bM_{\bar{a}b}(u,\bar{u})S(A_b)=\sum_bM_{\bar{a}b}(u,\bar{u})w_b=\overline{S(A_a)}.
	\end{equation}
	
	Moreover, we have
	
	\begin{prop}\label{conjcharge} Suppose $A\in C^{\infty}(M)\otimes\Omega^n_{\C^n\times M/M}/dF\wedge\Omega^{n-1}_{\C^n\times M/M}$ has (quasi-)homogeneous charge $l(A)$, then $\kappa(A)$ is also (quasi-)homogeneous with charge $l(\kappa(A))$. Furthermore,
		$$l(A)+l(\kappa(A))=nd.$$
	\end{prop}

	\begin{proof}
		It suffices to prove it for $\{A_a\}_{a=0}^{\mu-1}$. Let $w_a=S(A_a)$, i.e.
		$$w_a=A_a+\bar{\p}_F\nu_a \quad \text{for some $(n-1)$-form }\nu_a \text{ (with polynomial growth)}.$$
		Then
		\begin{equation}\label{bara}
		\bar{w}_a=\bar{A}_a+\p_F\bar{\nu}_a.
		\end{equation}
		First by Lemma \ref{Atheta}, (\ref{bara}), and $\P(\theta)\p_F=e^{i\theta}\p_F\P(\theta)$, proceed as what we did in Lemma \ref{uone}, one has \begin{equation}\label{bara1}\P(\theta)\bar{w}_a=e^{il_a\theta}\bar{w}_a=e^{il_a\theta}\sum_bM_{\bar{a}b}w_b.
		\end{equation}
		
		On the other hand,
		$$\bar{w}_a=\kappa(A_a)+\bar{\p}_F\tilde{\nu}_a \quad \text{for some $(n-1)$-form } \tilde{\nu}_a \text{ (with polynomial growth)},$$
		and $\kappa(A_a)=\sum_bM_{\bar{a}b}(u,\bar{u})A_b$, we have
		\begin{equation}\label{bara2}
		\P(\theta)\bar{w}_a=\sum_be^{i(nd-l_b)\theta}M_{\bar{a}b}w_b.
		\end{equation}
		(\ref{bara1})-(\ref{bara2}), one has
		\[\sum_b(e^{i(nd-l_b)\theta}-e^{il_a\theta})M_{\bar{a}b}w_b=0.\]
		Since $\{w_a\}$ is a basis, one has $M_{\bar{a}b}=0$ if $nd-l_b\neq l_a$, which implies that $\kappa(A_a)$ is homogeneous of degree $nd-l_a.$
	\end{proof}

	%


	\begin{rem}\label{n-2}
		If $f$ contains a unique isolated singularity at the origin, the strong nullstellensatz implies the existence of a sufficiently large $N$ such that $$z_i^N\in\lan \partial_{z_1}f,...,\partial_{z_n}f\ran$$ for all $i=1,2,...,n$. As a result, if $\phi$ is homogeneous with a sufficient large degree $K$, $\phi \in\lan \partial_{z_1}f,...,\partial_{z_n}f\ran$. In fact, Lemma \ref{Atheta} could provide an explicit description of $K$. To keep things simple, we will consider only the case where $f$ is homogeneous of degree $n$. A harmonic form's $U(1)$ charge cannot exceed $n(n-1)$ in this case. Thus, if $w=S(\phi \dz)$ for some homogeneous polynomial $\phi$ of degree $n(n-2)+1$, $w$ must have a $U(1)$ charge of $n(n-1)+1$. As a result, $w$ should be trivial, implying that $\phi\in \lan \partial_{z_1}f,...,\partial_{z_n}f\ran. $ As a result, we have $K=n(n-2)+1$ in this case.
	\end{rem}
	
	Now suppose $f\in\C[z_1,\ldots,z_n]$ is homogeneous of degree $n$. Then the above analysis implies that:
	
	\begin{prop-defn}[small $tt^*$ structure in the Landau-Ginburg B models \cite{Fan2020LGCYCB}]\label{smallLG}
		Let $f$ be a non-degenerate homogeneous degree $n$ polynomial and $F$ be its marginal deformation with parameter space $M$. Let $\H'\subset\H$ be the subbundle generated by $w_{k_a}=S(A_{k_a})$, where $l(A_{k_a})/n\in\mathbb{Z}.$ By Proposition \ref{conjcharge}, restriction of $\bar{\cdot},h,D,C,\bar{C}$ to $\H'$ defines a $tt^*$ structure, called small $tt^*$ geometry structure in the Landau-Ginzburg B models. Here $\bar{\cdot}$ denotes the complex conjugation.
	\end{prop-defn}
	
	\begin{proof} 
		Let $u=(u_1,...,u_s)$ be local coordinates of  $M$, and $\p_i=\p_{u_i} $.
		
		The restriction gives a well defined $tt^*$ structure, since
		\begin{enumerate}
			\item $D_i$ preserves the $U(1)$ charge of each $w_a$:\\
			Since $w_a$ is holomorphic, i.e. $\bar{D}_{\bar{i}}w_a=0$, we have
			$$\p_ih(w_a,w_b)=h(D_iw_a,w_b)+h(w_a,\bar{D}_{\bar{i}}w_b)=h(D_iw_a,w_b).$$
			By the non-degeneracy of $h$ and Corollary \ref{wawb}, we can see that $D_iw_a$ and $w_a$ have the same $U(1)$ charge.
			\item $C_i$ increases the $U(1)$ charge by $n$, since the operator $C_i$ is induced from multiplication by $\psi_i$ in $\Jac(F)$, where $\psi_i$ has the same $U(1)$ charge $n$ as $f$.
			\item $\mathcal{H}'$ is stable under $\bar{\cdot}$, by Proposition \ref{conjcharge}. Combining with (1) and (2), it implies that $\mathcal{H}'$ is stable under $\bar{D}_{\bar{i}}$ and $\bar{C}_{\bar{i}}$.
		\end{enumerate}
		
	\end{proof}
	
	It follows from the above analysis that one can also endow the following filtration on $\mathcal{H}'$
	$$\H'=\F^0\H'\supset\F^1\H'\supset\cdots\supset \F^{n-2}\H'$$
	with the filtration $\F^{k}\H'$ generated by
	$$\{w_{a}=S(A_a)~|~l_a\leq n(n-k-1)\}.$$
	It follows from the definition that under the isomorphism $S$, such filtration is the same as the pole order filtration on the rational differential forms after a degree shift.
	
	The above analysis implies that $\nabla^{(1,0)}=D+C$ satisfies the Griffiths transversality.
	
	Next, we denote $\Jac(F)'$ to be the subspace of $\Jac(F)$ that is generated by the monomial basis $\phi$, with $n|\deg(\phi)$. We have,
	$$\Jac(F)'\stackrel{\cong}{\longrightarrow}\mathcal{H}',\quad \phi\mapsto S(\phi dz_1\wedge\cdots\wedge dz_n).$$
	Denote $\dim\Jac(F)'=\mu'$. Moreover, we let $\phi_0=1$.
	
	\subsection{The vacuum line bundle and the Weil-Petersson metric}
	
	In this subsection, we assume that $f$ is homogeneous of degree $n$. Let us introduce the vacuum line bundle and Weil-Petersson metric in the LG models. Recall that $A_0=dz_1\wedge\cdots\wedge dz_n$ is the unique representative of lowest $U(1)$-charges $l_0=n$. As before, set $w_0=S(A_0)$ and $h_{0\bar{0}}=h(w_0,w_0)$. 
	Let $u=(u_1,...,u_s)$ be local coordinates of  $M$, and $\p_{i}=\p_{u_i} $.
	By Corollary \ref{wawb}, we have
	\begin{align*}
	&\p_ih(w_0,w_0)=h(D_iw_0,w_0)+h(w_0,\bar{D}_{\bar{i}}w_0)=h(D_iw_0,w_0),\\
	&0=\p_ih(w_0,w_a)=h(D_iw_0,w_a)+h(w_0,\bar{D}_{\bar{i}}w_a)=h(D_iw_0,w_a),\quad a>0.
	\end{align*}
	
	Hence, we have
	\begin{equation}\label{diazero}
	D_i w_0=\left(h_{\bar{0}0}^{-1}\p_ih_{0\bar{0}}\right)w_0,
	\end{equation}
	which means that the line subbundle $\mathcal{L}$ of $\H$ generated by $w_0$ is preserved by the Chern connection $D+\bar{D}$. It is called the vacuum line bundle in the physical literature.
	
	The Weil-Petersson metric is given by
	$$G_{i\bar{j}}=-\p_i\bar{\p}_{\bar{j}}\log h_{0\bar{0}}.$$

	\begin{thm} The Weil-Petersson metric $G_{i\bar{j}}$ is related to the $tt^*$ metric as follows:
		$$G_{i\bar{j}}=\frac{h(C_iw_0,C_jw_0)}{h(w_0,w_0)}.$$
	\end{thm}
	
	\begin{proof}
		First, since $w_0=S(dz_1\wedge\cdots\wedge dz_n)$ has the lowest $U(1)$ charge and $C_i$ increases the $U(1)$ charge by $n$, by a similar argument as in Lemma \ref{uone}
		\[h(w_0,C_iw_a)=0.\]
		While $h(\bar{C}_{\bar{i}}w_0,w_a)=h(w_0,C_iw_a)$
		implies that $\bar{C}_{\bar{i}}w_0=0.$
		
		Then,
		\begin{align*}
		\bar{\p}_{\bar{j}}\p_i\log h(w_0,w_0)&=-\frac{h(\bar{D}_{\bar{j}}D_iw_0,w_0)+h( D_iw_0,D_jw_0)}{h(w_0,w_0)}+\frac{h( D_iw_0,w_0)h(w_0,D_jw_0)}{h(w_0,w_0)^2}\\
		&=-\frac{h([\bar{D}_{\bar{j}},D_i]w_0,w_0)}{h(w_0,w_0)}-h_{\bar{0}0}^{-1}\p_ih_{0\bar{0}}\cdot\overline{h_{\bar{0}0}^{-1}\p_jh_{0\bar{0}}}+ h_{\bar{0}0}^{-1}\p_ih_{0\bar{0}}\cdot\overline{h_{\bar{0}0}^{-1}\p_jh_{0\bar{0}}}\\
		&=-\frac{h([C_i,\bar{C}_{\bar{j}}]w_0,w_0)}{h(w_0,w_0)}\\
		&=\frac{h( C_iw_0, C_jw_0)}{h(w_0,w_0)}.
		\end{align*}
	\end{proof}

	\section{Lefschetz Thimble, Riemann Bilinear Formula and Period Integral}\label{four}
	
	In this section, we express the $tt^*$ metric in terms of period integrals using the Riemann bilinear relation. Next, we present an explicit construction of the Lefschetz thimble in order to study the period integral.  Lastly, we investigate two types of period integrals.  We would like to point out that, Siebert-van Garrel-Ruddat \cite{Rud} also computes period integrals for Landau-Ginzburg models, which shows that, as predicted by mirror symmetry, the period integrals store enumerative information of the mirror dual.

	From now on, assume that $f$ is a non-degenerate homogeneous polynomial of degree $n$, $F$ is a marginal deformation of $f$ parametrized by $u\in M:$
	$$F(z,u)=f(z)+\sum_{i=1}^su_i\psi_i,$$
	where $\psi_i$ has the same degree as $f$.

	\subsection{Lefschetz thimble and intersection matrix}\label{not}
	
	Let $a>0$, and consider the sets
	$$F^{\geq a}=\{z\in\C^n\big|\Re(F(z,u))\geq a\},\quad F^{\leq -a}=\{z\in\C^n\big|\Re(F(z,u))\leq-a\}.$$
	By Morse theory, we know that if there are no critical values of $\Re(F(z,u))$ between $[a,b]$, then the set $F^{\leq a}$ is the deformation kernel of $F^{\leq b}$ by the flow generated by the vector field $\nabla\Re(F(z,u))/|\nabla\Re(F(z,u))|$. So if $a$ is large enough, there are no critical points in $F^{\leq-a}$ and each $F^{\leq-a}$ has the same homotopy type. We denote the equivalence class by $F^{-\infty}$. Similarly, we have $F^{+\infty}$. Then we can consider the relative homology groups $H_*(\C^n, F^{-\infty};\Z)$ and $H_*(\C^n,F^{+\infty};\Z)$.
	
	\begin{thm} $H_k(\C^n,F^{-\infty};\Z)$ is trivial except for $k=n$. $H_n(\C^n,F^{-\infty};\Z)$ is a free abelian group, and it is generated by the Lefschetz thimble $\{\gamma_k^-\}$. The same holds true for $H_*(\C^n, F^{+\infty};\Z)$.
	\end{thm}
	
	\begin{proof} See, for example, \cite{Fan2020LGCYCB}.
	\end{proof}
	
	Before moving on, let us give an explicit construction of the Lefschetz thimble $\gamma_k^-\in H_n(\C^n,F^{-\infty};\Z)$.
	
	Let $V_t$ be the level manifold $V_t=F^{-1}(t)$, then the hypersurface $X_{F}\subset\mathbb{P}^{n-1}$ could be regarded as the complement component of the projection of $V_t$:
	
	Consider $G(z_1,\ldots, z_n,z_{n+1})=F-tz_{n+1}^n$, $\bar{V}_t:=G^{-1}(0)\subset\mathbb{P}^n$, then $$X_F=\bar{V}_t-V_t.$$
	
	Let $c_t(s):=te^{2\pi i s}$, $0\leq s\leq1$, then $c_t$ induces a monodromy operator
	$$M_{c_t}: H_{n-1}(V_t)\to H_{n-1}(V_t).$$
	Following \cite{Fan2020LGCYCB}, we fix a basis $\{\sigma_k\}_{k=0}^{\mu-1}$ of $H_{n-1}(V_{-1})$, such that $\sigma_{k}\in\ker(M_{c_t}-\Id)$ for $0\leq k\leq\mu'-1$ (Recall that $\mu=\dim_{}\Jac(f),\mu'=\dim_{}\Jac(f)'$). Hence, there exists $\delta_k\in H_{n-2}(X_F)_0:=(H^{n-2}(X_F)_{\prim})^*$, such that $\sigma_k=\tau(\delta_k)$ for $0\leq k\leq\mu'-1$ (c.f., dualizing  \cite[(3.54)]{CECOTTI1991N}), where $\tau: H_{n-2}(X_F)\to H_{n-1}(\bar{V}_{-1}-X_F)=H_{n-1}(V_{-1})$ is the Leray coboundary map. For $t>0$, let $\Phi_t:V_{-t}\to V_{-1}$ be the map $(z_1,..,z_n)\to(t^{-\frac{1}{n}}z_1,...,t^{-\frac{1}{n}}z_n)$, and set
	\[\gamma_k^-:=\cup_{t>0}(\Phi_t)^*\sigma_k,\]
	then one can see that $\{\gamma_k^-\}_{k=0}^{\mu-1}$ is a basis of $H_n(\C^n,F^{-\infty};\mathbb{Z})$.
	
	Similarly, one can construct a basis $\{\gamma_k^+\}_{k=0}^{\mu-1}$ of $H_n(\C^n,F^{+\infty};\mathbb{Z}).$
	
	Now we would like to define intersection matrix between $H_n(\C^n,F^{-\infty};\mathbb{Z})$ and $H_n(\C^n,F^{+\infty};\mathbb{Z})$.
	
	First note that
	$$d_{2\Re(F)}=\bar{\p}_F+\p_F=e^{-F-\bar{F}}\circ d\circ e^{F+\bar{F}},\quad \Delta_{2\Re(F)}=2\Delta_F,$$
	then for any $\Delta_F$-harmonic form $w$, the modified forms
	$$e^{F+\bar{F}}w\quad\text{and}\quad e^{-F-\bar{F}}* w$$
	are $d$-closed. In particular, they give the representatives of $H^n(\C^n, F^{-\infty};\C)$ and $H^n(\C^n, F^{+\infty};\C)$, respectively (c.f. \cite{DY2020cohomology}).

	For each $\gamma_k^-\in H_n(\C^n,F^{-\infty};\mathbb{Z})$, the map
	$$w\mapsto \int_{\gamma_k^-}e^{F+\bar{F}}w,\quad w\in \ker{\Delta_{F(\cdot,u)}}$$
	is linear. Hence by Riesz representation, there exists $\alpha_k\in \ker{\Delta_{F(\cdot,u)}}$, such that
	\[\int_{\gamma_k^-}e^{F+\bar{F}}w=\int_{\C^n}w\wedge*\alpha_k=\int_{\C^n}e^{F+\bar{F}}w\wedge e^{-F-\bar{F}}*\alpha_k.\]

	Then set
	$$PD(\gamma_k^-):=e^{-F-\bar{F}}*\alpha_k\in H^n(\C^n,F^{+\infty};\C).$$
	Similarly, we can define $PD(\gamma_k^+),$ such that
	\[\int_{\gamma_k^+}e^{-F-\bar{F}}*w=\int_{\C^n}PD(\gamma^+_k)\wedge e^{-F-\bar{F}}*w\]
	for all $w\in\ker{\Delta_{F(\cdot,u)}}.$
	\begin{defn}\label{intermatrix}
		The intersection matrix $I=(I_{kl})_{\mu\times\mu}$ is defined by $$I_{kl}=\int_{\gamma_k^-}PD(\gamma_l^+)=\int_{\C^n}PD({\gamma}_l^+)\wedge PD({\gamma}_k^-).$$
		for $0\leq k,l\leq \mu-1.$

		While $I'$ is defined as a submatrix of $I$, that is $I_{kl}'=I_{kl}$ for $0\leq k,l\leq \mu'-1.$
		Now set $\I=I^{-1},\I'=(I')^{-1}.$
	\end{defn}
	
	\begin{rem}
		While there is a purely topological way to define the intersection matrix of the Lefschetz thimble, the definition provided above is more convenient for our purposes, since in our setting, Proposition \ref{riebil} below is trivial.
	\end{rem}
	By definition, the forms $PD(\gamma_k^-)$ and $PD(\gamma_k^+)$ depend on $u$, however,
	\begin{prop}
		The intersection matrix $I$ is locally constant, i.e. when $|u|$ is small, $\p I=\bar{\p} I=0.$
	\end{prop}
	\begin{proof}
		Suppose $w$ is $\Delta_{F(\cdot,u)}$-harmonic, then if $|u|$ is small, by Agmon estimate (Lemma \ref{epde12}), 
		$$w':=e^{(F(\cdot,u)-f)+\overline{F(\cdot,u)-f}}w$$ 
		has exponential decay and is $d_{2\Re(f)}$-closed. Then proceed as what we did in \cite[\S 7.5]{DY2020cohomology}, we can find $\nu$, s.t. $w'-\Pi_0 w'=d_{2\Re(f)}\nu$, and $\nu$ has exponential decay. Here $\Pi_0$ is the projection $L^2\to \ker{\Delta_f}$. Hence,
		\[\int_{\gamma_k^-}e^{F+\bar{F}}w=\int_{\gamma_k^-}e^{f+\bar{f}}w'=\int_{\gamma_k^-}e^{f+\bar{f}}\Pi_0w'.\]
		Again, as what we did in \cite[\S 7.5]{DY2020cohomology},
		$$\Pi_0\circ e^{(F(\cdot,u)-f)+\overline{F(\cdot,u)-f}}:H_{(2)}^*(\C^n,d_{2\Re(F)(\cdot,u)})\to H_{(2)}^*(\C^n,d_{2\Re(f)})$$ 
		is isomorphic, for the similar reason, one can see that $PD(\gamma_k^-)(u,\bar{u})-PD(\gamma_k^-)(0)=d \mu$ for some differential form $\mu$ with exponential decay on $\Re(f)<0.$
		
		Hence, $I$ is locally constant.
	\end{proof}

	By the definition of intersection matrix, the following Riemann bilinear formula is straightforward:
	\begin{prop}\label{riebil}
		If $w$ and $w'$ are harmonic, then
		$$\int_{\C^n}w\wedge*w'=\sum_{k,l}\left(\int_{\gamma_k^-}e^{F+\bar{F}}w\right)(\I)_{kl}\left(\int_{\gamma_l^+}e^{-F-\bar{F}}w'\right).$$
	\end{prop}

	\subsection{Two types of integrals: $\int_{\gamma^-}e^{F}A $ vs $\int_{\gamma^-}e^{F+\bar{F}}w$}\label{periodintegral}
	
	This subsection is dedicated to compare the two kinds of period integrals in the LG B-model
	$$\int_{\gamma^-}e^{F}A \quad\text{and}\quad\int_{\gamma^-}e^{F+\bar{F}}w,$$
	i.e., to prove Theorem \ref{lgcy1} (or Theorem \ref{132} below).
	
	Recall that we fix  a monomial representation $\{\phi_a\}_{a=0}^{\mu-1}$ of basis of $\Jac(f)$, where $\mu=\dim\Jac(f)$ and $\phi_0=1$. Moreover, we assume that $\w(\phi_a)\leq \w(\phi_b)$ if $a<b$ (See Definition \ref{quasi} for the definition of $\w$). 
	
Since $|u|$ is small, $\{\phi_a\}_{a=0}^{\mu-1}$ is still a representation of basis of $\Jac(F(\cdot,u)).$ Moreover, $A_a=\phi_adz_1...dz_n, w_a=S(A_a)$
	
	
	Recall that $\{e^{F+\bar{F}}w_a\}$ and  $\{e^{-F-\bar{F}}*w_a\}$ give a basis of $H^n(\C^n, F^{-\infty};\C)$ and $H^n(\C^n, F^{+\infty};\C)$ respectively, we consider
	$$\A_{ka}=\int_{\gamma_k^-}e^{F+\bar{F}}w_a, \quad \tilde{\A}_{ka}=\int_{\gamma_k^+}e^{-F-\bar{F}}* w_a.$$

	
	Since $\{e^FA_a\}$ and $\{e^{-F}* A_a\}$
	are also $d$-closed $n$-forms, which also give a basis of $H^n(\C^n, F^{-\infty};\C)$ and $H^n(\C^n, F^{+\infty};\C)$ respectively. We can also consider the following integral
	$$\B_{ka}=\int_{\gamma_k^-}e^FA_a,\quad \tilde{\B}_{ka}=\int_{\gamma_k^+}e^{-F}* A_a.$$

	Considering the pairing of these two types of integrals, one should expect the following Riemann bilinear relations.
	
	\begin{prop}\label{riebil1}
		If $w$ is harmonic, then
		$$\int_{\C^n}e^{-\bar{F}}A_a\wedge*w=\sum_{k,l}\left(\int_{\gamma_k^-}e^{F}A_a\right)(\I)_{kl}\left(\int_{\gamma_l^+}e^{-F-\bar{F}}w\right).$$
	\end{prop}
	Proposition \ref{riebil1} is actually nontrivial. We will prove this proposition in the appendix.

	Since $\A$ and $\widetilde{\A}$ are invertible for $u\in M$, there exist matrix-valued functions $\T$ and $\widetilde{\T}$, such that $\B=\T\A$ and $\widetilde\B=\widetilde\T\tilde\A.$ More specifically,
	
	\[\int_{\gamma_k^-}e^{F}A_a=\sum_a(\T)_{ab}\int_{\gamma_k^-}e^{F+\bar{F}}w_b.\]

	
	\begin{thm}
		\label{132}
		$\T$ and $\tilde{\T}$ are (block) unit lower triangular matrix. More explicitly,
		\begin{align*}
		(\T)_{aa}&=1 \quad\mbox{ for $0\leq a\leq \mu-1$};\\
		(\T)_{ab}&=0 \quad\mbox{ if $a<b$};\\
		(\T)_{ab}&=\begin{cases}
		&0, \quad\mbox{ if $\frac{l(A_a)-l(A_b)}{n}\notin \mathbb{Z}^+$};\\
		&\frac{(-1)^{l_{ab}}}{l_{ab}!}\sum_{l_c=l_b}\int_{\C^n}\bar{F}^{l_{ab}}A_a\wedge*\bar{w}_c h^{cb}, \quad\mbox{ if $l_{ab}:=\frac{l(A_a)-l(A_b)}{n}\in \mathbb{Z}^+$.}
		\end{cases}
		\end{align*}
		
		Here $(h^{ab})$ is the inverse of $(h_{ab})$, $h_{ab}:=h(w_a,w_b).$
		
		In particular,
		\[\int_{\gamma_k^-}e^{F+\bar{F}}w_0=\int_{\gamma_k^-}e^Fdz_1\wedge\cdots\wedge dz_n,\]
		\[\int_{\gamma_k^+}e^{-F-\bar{F}}*w_0=\int_{\gamma_k^+}e^{-F}*dz_1\wedge\cdots\wedge dz_n,\]
	\end{thm}

	\begin{proof}
		By Riemann bilinear formula (Proposition \ref{riebil1}), on the one hand,
		\begin{equation}\label{rieb1}
		\left(\int_{\gamma_k^-}e^{F}A_a\right)\I_{kl}\overline{\left(\int_{\gamma_l^+}e^{-F-\bar{F}}*{w}_c\right)}=\int_{\C^n}e^{-\bar{F}}A_a\wedge*\overline{w_c}=h(e^{-\bar{F}}A_a,w_c).
		\end{equation}
		On the other hand,
		\begin{equation}\label{rieb2}
		\begin{aligned}
		&\int_{\gamma_k^-}e^FA_a\I^{kl}\overline{\left(\int_{\gamma_l^+}e^{-F-\bar{F}}* w_c\right)}\\
		=~&\left(\sum_b\T_{ab}\int_{\gamma_k^-}e^{F+\bar{F}}w_b\right)\I_{kl}\overline{\left(\int_{\gamma_l^+}e^{-F-\bar{F}}* w_c\right)}\\
		=~&\sum_b\T_{ab}\int_{\C^n}w_b\wedge*\overline{w_c}\\
		=~&\sum_{b:l_b=l_c}\T_{ab}h(w_b,w_c).
		\end{aligned}
		\end{equation}
		
		Let $t=e^{i\theta}$, then
		\begin{equation} \label{htheta1}h(e^{-\bar{F}}A_a,w_c)=h(\P(\theta)e^{-\bar{F}}A_a,\P(\theta)w_c)=t^{l_c-l_a}h(e^{-t^n\bar{F}}A_a,w_c).
		\end{equation}
		
		\begin{itemize}
			\item If $\frac{l_c-l_a}{n}\notin \Z$, then take $t$ to be the $n$-th root of $1$, such that $t^{l_c-l_a}\neq1$, by (\ref{htheta1}), we see that
			$$h(e^{-\bar{F}}A_a,w_c)=0.$$
			This together with (\ref{rieb1}), (\ref{rieb2}) and the fact that the matrix $(h(w_b,w_c))_{l_b=l_c}$ is invertible imply that
			$$(\T)_{ab}=0,\quad\text{if }\frac{l_a-l_b}{n}\notin\Z.$$
			\item If $l_c-l_a=nl_{ca}$ for some integer $l_{ca}$, then we consider the function
			$$v_{ac}(s)=h(e^{-s\bar{F}}A_a,w_c).$$
			By Lemma \ref{Agmonhol} below, it is holomorphic in $s$ when $|s|<2$. Then
			\begin{itemize}
				\item if $l_{ca}>0$: when $|s|=1$, by (\ref{htheta1}) $s^{l_{ca}}v_{ac}(s)=v_{ac}(1)$, which implies $v_{ac}(s)=s^{-l_{ca}}v_{ac}(1)$ when $|s|=1$. However, by the properties of holomorphic function, we must have $v_{ac}(s)=s^{-l_{ca}}v_{ac}(1)$ when $0<|s|<2.$ Since $v_{ac}$ is holomorphic function on $|s|<2$, we must have $v_{ac}(1)=h(e^{-\bar{F}}A_a,w_c)=0.$
				$$\T_{ab}=0,\quad \mbox{ if } l_b=l_c.$$
				\item if $l_{ca}\leq0$: if follows from the same argument as above that $v_{ac}(s)=s^{-l_{ca}}v_{ac}(1)$ when $|s|<2$. Applying $\left(\frac{d}{ds}\right)^{-l_{ca}}\bigg|_{s=0}$ to both sides, we obtain
				$$h(e^{-\bar{F}}A_a,w_c)=\frac{1}{(-l_{ca})!}\int_{\C^n}\bar{F}^{-l_{ca}}A_a\wedge*\bar{w}_c.$$
				
				In particular, when $l_{ca}=0$, $$h(e^{-\bar{F}}A_a,w_c)=h(A_a,w_c)=h(w_a,w_c)-h(\bar{\p}_{F(\cdot,u)}\nu_a,w_c)=h(w_a,w_c).$$
			\end{itemize}
		\end{itemize}
		
		As a result, Theorem \ref{132} follows.
		
	\end{proof}

	\begin{lem}\label{Agmonhol}
		The function $v_{ac}(s)=\int_{\C^n}e^{-s\bar{F}}A_a\wedge*\bar{w}_c$ is holomorphic on the disk $|s|<2.$
	\end{lem}
	\begin{proof}
		Without loss of generality, it suffices to prove it when $u=0$. Note that $$\{\bar{\p}_{f},\bar{\p}_{f}^{\dagger}\}=\frac{1}{2}\{d_{2\Re(f)},d_{2\Re(f)}^{\dagger}\}.$$ 
		By Agmon estimate, there exists $C_b>0$ for any $b\in(0,1)$, such that
		\[|w_c|\leq C_be^{-b\rho}\|w_c\|_{L^2},\]
		where $\rho(z)$ is the Agmon distance between $z$ and $0$ with respect to Agmon metric $2|\nabla \Re(f)|^2g_0$ (See appendix for more details). It follows from the properties of Agmon distance and holomorphic functions that $\rho(z)\geq 2|\Re(e^{i\theta}f)|$ for any $\theta\in\R$. As a result, for any $a<b$,
		if $|s|\leq 2a$,
		\begin{align*}
		\int_{\C^n}e^{-s\bar{f}}A_a\wedge*\bar{w}_c&\leq C\int_{\C^n} e^{2a\Re(\frac{sf}{|s|})}e^{-b\rho}\mathrm{dvol}_{\C^n}\leq C\int_{\C^n} e^{-(b-a)\rho}\mathrm{dvol}_{\C^n}<\infty.
		\end{align*}
		Hence $v_{ac}(s)$ is holomorphic when $|s|<2.$
	\end{proof}

	\section{Calabi-Yau/Landau-Ginzburg Correspondence}\label{five}

	In this section, we assume that
	$$f:\C^n\rightarrow\C\quad n\geq3$$
	is a non-degenerate homogeneous polynomial of degree $n$. Regarding $[z_1,z_2,\ldots,z_n]$ as the homogeneous coordinate in $\mathbb{P}^{n-1}$, the zero set of $f$, denoted by $X_{f}$, defines a Calabi-Yau hypersurface in $\mathbb{P}^{n-1}$.

	Now consider the marginal deformation $F=f+\sum_{i=1}^s u_i\psi_i$ of $f$, then for $|u|$ small enough, $F$ determines a family of Calabi-Yau hypersurface $X_{F(\cdot,u)}$ parametrized by $u\in M$, which gives a deformation of $X_f$.

	\subsection{Residue maps}

	\def\res{\operatorname{res}}
	For each $\phi_a\in \Jac(F)'$, $\frac{l_a}{n}\in \Z^+$. Then set $n_a:=\frac{l_a}{n}$. Let
	$$\Omega:=\sum_{i=1}^n(-1)^{i+1}z_{i} d z_{1} \wedge \ldots \widehat{dz_i}\ldots \wedge d z_{n},$$
	and denote $$\Omega_a:=\frac{\phi_a\Omega}{F(\cdot,u)^{n_a}}.$$ Then one can see that $\Omega_a\in H^{n-1}(\mathbb{P}^{n-1}-X_{F(\cdot,u)})$.
	
	Let $\tau: H_{n-2}(X_{F(\cdot,u)})\to H_{n-1}(\mathbb{P}^{n-1}-X_{F(\cdot,u)})$ be the Leray coboundary map, one can define the residue map $\res:H^{n-1}(\mathbb{P}^{n-1}-X_{F(\cdot,u)})\to H^{n-2}(X_{F(\cdot,u)})$ by
	\[2\pi i\int_{\delta}\res(\Omega_a)=\int_{\tau(\delta)}\Omega_a.\]

	It was shown that (see, for example, \cite{carlson1980infinitesimal})
	
	\begin{prop} The residue map
		$$\res: H^{n-1}(\mathbb{P}^{n-1}-X_{F(\cdot,u)},\mathbb{C})\longrightarrow H^{n-2}(X_{F(\cdot,u)},\C)$$
		is injective, with the image being the $(n-2)$-th primitive cohomology of $X_{F(z,u)}$. Moreover, the residue map preserves the Hodge structure. That is,
		\begin{equation}\label{fil}\res(\Omega_a)\in \F^{n-n_a-1}H^{n-2}(X_{F(\cdot,u)})_{\prim}.
		\end{equation}
		In particular, $\res(\Omega_0)$ is a nowhere vanishing holomorphic $(n-2,0)$ form on $X_{F(\cdot,u)}$.
	\end{prop}

	\begin{defn}\label{resmap}
		We define the map $R_u:\Jac(F)'\to H^{n-2}(X_F)$ via
		\[R_u(\phi_a):=\res\left(\Omega_a\right).\]
		In particular, $R_u(1)$ is a nowhere vanishing holomorphic $(n-2,0)$ form on $X_F$.
	\end{defn}



	\subsection{The correspondence between the period integrals}
	
	By Theorem \ref{132}, we also have
	\begin{equation}\label{haref}\int_{\gamma_k^-}e^{F+\bar{F}}w_a=\int_{\gamma_k^-}\left(e^FA_a-\sum_{b:l_b<l_a}(\T)_{ab}(u,\bar{u})e^FA_b\right).\end{equation}
	\begin{defn}
		We define the map $r_u:\H'\to H^{n-2}(X_F)_{\prim}$ via
		\[r_u(w_a)=2\pi i\left((-1)^{n_a-1}(n_a-1)!R_u(\phi_a)-\sum_{b:l_b<l_a}(\T)_{ab}(u,\bar{u})(-1)^{n_b-1}(n_b-1)!R_u(\phi_b)\right).\]
	\end{defn}
	\begin{rem}\label{rsha}
		By (\ref{fil}), one can see that if $l(A)=l$ and $\frac{l}{n}\in\Z$, then
		$$r_u(S(A))\in \F^{n-1-\frac{l}{n}}H^{n-2}(X_F)_{\prim}.$$
	\end{rem}
	
	\begin{lem}\label{ruan}
		For $0\leq k\leq \mu'-1$, $\phi_a\in\Jac(F)'$, $A_a=\phi_adz_1\wedge\cdots\wedge dz_n$,
		\[\int_{\gamma_k^-}e^{F}A_a=2\pi i(-1)^{n_a-1}(n_a-1)!\int_{\delta_k}R_u(\phi_a),\]
		and
		\[\int_{\gamma_k^+}e^{-F}*A_a=2\pi i(-1)^{n_a}(n_a-1)!\int_{\delta_k}*R_u(\phi_a).\]
	\end{lem}
	\begin{proof}
		It suffices to compute at $u=0.$
		
		Let $\sigma_k(t)=(\Phi_T)^*\sigma_k\in H_n(V_{-t}).$ For each $t\neq 0$, let
		$$\tilde{c}_t(\theta):=-t+\frac{te^{i\theta}}{2},\quad T(\sigma_k)(t):=\cup_{\theta}P_{\theta}(\sigma_k(-\frac{t}{2})),$$
		where $P_{\theta}: H_{n}(V_{-\frac{t}{2}})\to H_n(V_{\tilde{c}_t(\theta)})$ is the parallel transport along $\tilde{c}_t.$
		Then
		\begin{align}
		\begin{split}
		\int_{\gamma_k^-}e^{f}A_a&=\int_0^\infty e^{-t}\int_{\sigma_k(t)}\frac{A_a}{df}dt\\
		&=\frac{1}{2\pi i}\int_0^\infty e^{-t}\int_{T(\sigma_k(t))}\frac{A_a}{f+t}dt.
		\end{split}
		\end{align}
		
		Let $v(t)=\int_{T(\sigma_k(t))}\frac{A_a}{f+t},$ then $v(t)=t^{n_a-1}v(1)$ (recall that $n_a=\frac{l_a}{n}$), proceed as in the proof of \cite[Lemma A.2]{chiodo2014landau}, and take the differential of $t$, one has
		\[v(t)=(-1)^{n_a-1}t^{n_a-1}\int_{T(\sigma_k(t))}\frac{A_a}{(f+t)^{n_a}}\]
		
		It follows from  \cite[Theorem 4.2]{Fan2020LGCYCB} that
		\[\int_{T(\sigma_k(t))}\frac{A_a}{(f+t)^{n_a}}=\lim_{t\to 0}\int_{T(\sigma_k(t))}\frac{A_a}{(f+t)^{n_a}}=(2\pi i)^2\int_{\delta_k}R_0(\phi_a).\]
	\end{proof}
	
	We show that the real structure and the bi-grading are preserved under $r$:
	\begin{thm}\label{rpqdegree}
		For $0\leq k\leq\mu'-1$, one has
		\begin{equation}\label{periodlgcy}\int_{\gamma_k^-}e^{F+\bar{F}}w=\int_{\delta_k} r_u(w)\end{equation}
		and \begin{equation}\label{periodlgcy0}\int_{\gamma_k^+}e^{-F-\bar{F}}*w=\int_{\delta_k} *r_u(w)\end{equation}
		for any $w\in\H'.$
		
		Moreover, $r_u(\bar{w}_a)=\overline{r_u(w_a)}$ represent the same class in $H^*(X_F)$ and $r_u(w_a)\in H^{n-n_a-1,n_a-1}(X_F)_{\prim}$.
	\end{thm}
	\begin{proof}
		By Lemma \ref{ruan} and (\ref{haref}), for $0\leq k\leq\mu'-1$, any $w\in\H'$,
		\begin{equation}\label{ruan11}
		\int_{\gamma_k^-}e^{F+\bar{F}}w=\int_{\delta_k} r_u(w).
		\end{equation}
		
		In particular, take $w=\bar{w}_a$, by (\ref{ruan11}),
		
		\begin{equation}\label{ruan111}
		\int_{\gamma_k^-}e^{F+\bar{F}}\bar{w}_a={\int_{\delta_k} r_u(\bar{w}_a)}.
		\end{equation}
		
		Moreover,
		\begin{equation}\label{ruan112}
		\int_{\gamma_k^-}e^{F+\bar{F}}\bar{w}_a=\overline{\int_{\gamma_k^-}e^{F+\bar{F}}{w}_a}=\overline{\int_{\delta_k} r_u({w}_a)}={\int_{\delta_k} \overline{r_u({w}_a)}}.
		\end{equation}
		
		Thus,(\ref{ruan111}) and (\ref{ruan112}) and Remark \ref{rsha} implies that \begin{equation}\label{ruan1}r_u(\bar{w}_a)=\overline{r_u(w_a)}\in \overline{\F^{n-n_a-1}H^{n-2}(X_F)_{\prim}},\end{equation}
		where $n_a:=\frac{l_a}{n}.$
		
		One the other hand, since $\bar{w}_a=S(\kappa(\phi_adz_1\wedge\cdots\wedge dz_n))$, by Remark \ref{rsha},
		$$r_u(\bar{w}_a)\in \F^{n_a-1}H^{n-2}(X_F)_{\prim}.$$
		As a result,
		\begin{equation} \label{rbwa}
		r_u(\bar{w}_a)\in \overline{\F^{n-n_a-1}H^{n-2}(X_F)_{\prim}}\cap \F^{n_a-1}H^{n-2}(X_F)_{\prim}=H^{n_a-1,n-n_a-1}(X_F)_{\prim}.
		\end{equation}
		
		(\ref{ruan1}) and (\ref{rbwa}) tell us that $r_u(w_a)\in H^{n-n_a-1,n_a-1}(X_F)_{\prim}.$
	\end{proof}

	\begin{thm}\label{fanlan}
		If $\phi_a,\phi_b\in \Jac(F)'$, then
		$$\int_{\C^n}w_a\wedge*w_b=\frac{(2\pi)^2}{2^{n-2}}\int_{X_F}r_u(w_a)\wedge*r_u(w_b).$$
	\end{thm}
	\begin{proof}
		By Lemma \ref{wawb} and Lemma \ref{conjcharge}, if $l_a+l_b\neq n^2$, we have
		$$\int_{\C^n}w_a\wedge*w_b=0.$$
		While by Proposition \ref{rpqdegree}, if $l_a+l_b\neq n^2$, we also have
		$$\int_{X_F}r_u(w_a)\wedge*r_u(w_b)=0.$$
		
		Hence it suffices to prove the case of $l_a+l_b=n^2$. By (\ref{fil}), Theorem \ref{rpqdegree} and the definition of $r_u$ and $R_u$
		\be \label{last1}r_u(w_a)=(2\pi i)(-1)^{n_a-1}(n_a-1)!R_u(\phi_a)^{n-n_a-1,n_a-1},\ee
		where $R_u(\phi_a)^{n-n_a-1,n_a-1}$ denotes the $(n-n_a-1,n_a-1)$ component of $R_u(\phi_a).$
		
		By (\ref{last1}),  \cite[Theorem 3.4]{shen2016explicit} and Proposition \ref{carlson},
		\begin{align*}&\ \ \ \ \int_{\C^n}w_a\wedge*w_b\\
		&=\frac{(-1)^{(n_a-1)(n_b-1)+1}(n_a-1)!(n_b-1)!}{2^{n-2}}\int_{X_F}R_u(\phi_a)^{n-n_a-1,n_a-1}\wedge*R_u(\phi_b)^{n-n_b-1,n_b-1}\\
		&=\frac{(2\pi)^2}{2^{n-2}}\int_{X_F} r_u(w_a)\wedge*r_u(w_b).\end{align*}
	\end{proof}
	
	
	\subsection{The correspondence for intersection matrices}
	Here we use the same notation as in \cref{not}.

	If $\phi_a\in \Jac(F)'$, for $k\geq\mu'$,
	\[\int_{\gamma_k^-}e^{F}A_a=\int_0^\infty e^{-t}\int_{\sigma_k}\frac{A_a}{dF},\]
	where $\frac{A_a}{dF}$ is the Gelfand-Leray form (c.f. \cite[Lemma 10.3]{arnold2012singularities}).
	
	Integration by substitution tells us that
	$$\int_{\sigma_k}\frac{A_a}{dF}=0,\quad \text{if } M_{c_t}\sigma_k\neq\sigma_k,$$
	which implies if $k\geq \mu'$
	\begin{equation}\label{monovan}\int_{\gamma_k^-}e^{f}A_a=0\end{equation}
	
	\begin{thm}\label{intersec}
		$\I'=\frac{\pi^2}{2^{n-4}}\I^{CY}.$
	\end{thm}
	\begin{proof}
		On the one hand, by (\ref{monovan}) and Riemann bilinear formula,
		\begin{align*}
		\int_{\C^n}w_a\wedge*w_b&=\sum_{0\leq k,l< \mu'}\left(\int_{\gamma_k^-}e^{F+\bar{F}}w_a\right)(\I')_{kl}\left(\int_{\gamma_l^+}e^{-F-\bar{F}}*w_b\right)
		\\
		&=\sum_{0\leq k,l< \mu'}\left(\int_{\delta_k}r_u(w_a)\right)(\I')_{kl}\left(\int_{{\delta}_l}*r_u(w_b)\right)\mbox{ (By (\ref{haref}) and Lemma \ref{ruan})}.
		\end{align*}
		
		On the other hand, by Theorem \ref{fanlan},
		
		\begin{align*}
		\int_{\C^n}w_a\wedge*w_b&=\frac{\pi^2}{2^{n-4}}\int_{X_F}r_u(w_a)\wedge*r_u(w_b)\\
		&=\frac{\pi^2}{2^{n-4}}\sum_{0\leq k,l< \mu'}\left(\int_{\delta_k}r_u(w_a)\right)(\I^{CY})_{kl}\left(\int_{{\delta}_l}*r_u(w_b)\right).
		\end{align*}
		
		Hence, $\I'=\frac{\pi^2}{2^{n-4}}\I^{CY}.$
	\end{proof}
	
	\begin{proof}[Proof of Theorem \ref{133}]
		By Theorem \ref{intersec} and Lemma \ref{ruan}, it's easy to see that
		\[\int_{\C^n}w_0\wedge*\bar{w}_0=\frac{\pi^2}{2^{n-4}}\int_{X_F}r_u(w_0)\wedge *\overline{r_u(w_0)}.\]
		
		Then $G_{i\bar{j}}=G_{i\bar{j}}^{CY}$ follows from the following:
		$$G_{i\bar{j}}=-\p_i\bar{\p}_{\bar{j}}\log(\int_{\C^n}w_0\wedge*\bar{w}_0),\quad G^{CY}=-\p_i\bar{\p}_{\bar{j}}\log(\int_{X_F}r_u(w_0)\wedge *\overline{r_u(w_0)}).$$
	\end{proof}
	
	To show Calabi-Yau/Landau-Ginzburg correspondence for $tt^*$ geometry, we first show that:
	\begin{lem}\label{gmc}
		Let $u=(u_1,...,u_s)$ be local coordinates of  $M$, and $\p_i=\p_{u_i}$.
		\[\partial_i\int_{\gamma_k^-}e^{F+\bar{F}}w_a=\int_{\gamma_k^-}e^{F+\bar{F}}(D_{i}+C_{i})w_a,\]
		\[\bar{\p}_{\bar{i}}\int_{\gamma_k^-}e^{F+\bar{F}}w_a=\int_{\gamma_k^-}e^{F+\bar{F}}(\bar{D}_{\bar{i}}+\bar{C}_{\bar{i}})w_a.\]
	\end{lem}
	\begin{proof}
		First, notice that
		\begin{equation}\label{dplusc1}\partial_i\int_{\gamma_k^-}e^{F+\bar{F}}w_a=\int_{\gamma_k^-}\partial_i(e^{F+\bar{F}}w_a)=\int_{\gamma_k^-}e^{F+\bar{F}}(\partial_i+\psi_i)w_a.\end{equation}
		
		Since $e^{F+\bar{F}}w_a$ is $d$-closed for all $u$ and $\partial_{i}$ commutes with $d$, $e^{F+\bar{F}}(\partial_i+\psi_i)w_a=\partial_{i}(e^{F+\bar{F}}w_a)$ is also $d$-closed, which implies that $(\partial_i+\psi_i)w_a$ is $d_{2\Re(F)}$-closed. Proceed as in the proof of Proposition \ref{riebil1}, one shows that \[(\partial_i+\psi_i)w_a=\Pi_u((\partial_i+\psi_i)w_a)+d_{2\Re(F)}\alpha=(D_i+C_i)w_a+d_{2\Re(F)}\alpha\]
		for some differential form $\alpha$ with exponential decay.
		
		As a result,
		\begin{equation}\label{dplusc2}\int_{\gamma_k^-}e^{F+\bar{F}}(\partial_i+\psi_i)w_a
		=\int_{\gamma_k^-}e^{F+\bar{F}}(D_{i}+C_{i})w_a.
		\end{equation}
		
		By (\ref{dplusc1}) and (\ref{dplusc2}), one has
		\[\partial_i\int_{\gamma_k^-}e^{F+\bar{F}}w_a=\int_{\gamma_k^-}e^{F+\bar{F}}(D_{i}+C_{i})w_a.\]
		
		Similarly, one can show that
		\[\bar{\p}_{\bar{i}}\int_{\gamma_k^-}e^{F+\bar{F}}w_a=\int_{\gamma_k^-}e^{F+\bar{F}}(\bar{D}_{\bar{i}}+\bar{C}_{\bar{i}})w_a.\]
	\end{proof}
	
	\begin{proof}[Proof of Theorem \ref{134}]
		First, let $r_u'=2^{\frac{4-n}{2}}\pi r_u.$
		By Theorem \ref{rpqdegree}, the real structure is preserved by $r_u'$. By  Theorem \ref{fanlan}, the bilinear pairing is also preserved by $r_u'.$ Hence, if we have
		\begin{flalign}
		\begin{split}\label{gaussmanin}
		r_u'(Dw_a)&=D^{CY}r_u'(w_a),\quad r_u'(\bar{D}w_a)=\bar{D}^{CY}r_u'(w_a),\\
		r_u'(Cw_a)&=C^{CY}r_u'(w_a),\quad r_u'(\bar{C}w_a)=\bar{C}^{CY}r_u'(w_a)
		\end{split}
		\end{flalign}
		for $\phi_a\in\Jac(F)',$
		$r_u':\H'\to H^{n-2}(X_F)_{\prim}$ induces an isomorphism of $tt^*$ structure. 
		
		Now we prove (\ref{gaussmanin}). By (\ref{periodlgcy}), one has
		\[\partial_i\int_{\gamma_k^-}e^{F+\bar{F}}w=\partial_i\int_{\delta_k} r_u'(w),\]
		\[\bar{\p}_{\bar{i}}\int_{\gamma_k^-}e^{F+\bar{F}}w=\bar{\p}_{\bar{i}}\int_{\delta_k} r_u'(w).\]
		Hence by Lemma \ref{gmc}, Theorem \ref{rpqdegree} and the definition of Gauss-Manin connection $D^{CY}+C^{CY}$,
		\begin{equation}\label{gaussmanin1}
		\begin{aligned}
		&r_u'\left((D+C)w_a\right)=(D^{CY}+C^{CY})r_u'(w_a),\\
		&r_u'\left((\bar{D}+\bar{C})w_a\right)=(\bar{D}^{CY}+\bar{C}^{CY})r_u'(w_a).
		\end{aligned}
		\end{equation}
		
		Suppose $w_b=S(A_b)$ such that $l(A_b)=l(A_a)$, then
		by Riemann bilinear formula and Theorem \ref{rpqdegree},
		\[\int_{\C^n} w_a\wedge*\bar{w}_b=\int_{X_F}r_u'(w_a)\wedge*\overline{r_u'(w_b)}.\]
		As a result,
		\[\partial_i\int_{\C^n} w_a\wedge*\bar{w}_b=\partial_i\int_{X_F}r_u'(w_a)\wedge*\overline{r_u'(w_b)}.\]
		
		Then notice that
		\[\partial_i\int_{\C^n} w_a\wedge*\bar{w}_b=\int_{\C^n} D_iw_a\wedge*\bar{w}_b,\]
		and
		\[\partial_i\int_{X_F}r_u'(w_a)\wedge*\overline{r_u'(w_b)}=\int_{X_F}D^{CY}_ir_u'(w_a)\wedge*\overline{r_u'(w_b)}.\]
		
		As a result,
		\[r_u'(Dw_a)=D^{CY}r_u'(w_a).\]
		Similarly,
		\[r_u'(\bar{D}w_a)=\bar{D}^{CY}r_u'(w_a).\]
		
		Together with (\ref{gaussmanin1}), we have (\ref{gaussmanin}).

	\end{proof}
	
	Finally, we summarize our results. Let $F$ be the marginal deformation of a non-degenerate homogeneous polynomial $f$ of degree $n$ on $\C[z_1,\ldots,z_n]$. The (small) $tt^*$ structure in the Calabi-Yau B-model and Landau-Ginzburg B-model are given by
	\begin{itemize}
		\item $tt^*$ structure in the Calabi-Yau B model:
		$$\left(H^{n-2}(X_F)_{\prim}, \bar{\cdot},D^{CY}, \bar{D}^{CY}, C^{CY},\bar{C}^{CY},g\right).$$
		\item small $tt^*$ structure in the Landau-Ginzburg B model:
		$$\left(\mathcal{H}',\bar{\cdot}, D,\bar{D}, C,\bar{C},h\right).$$
		Here $\bar{\cdot}$ denotes the complex conjugations.
	\end{itemize}
	The map $r_u':\mathcal{H}'\rightarrow H^{n-2}(X_F)_{\prim}$ is an isomorphism of two $tt^*$ geometry structures above. In particular, $r_u'$ not only preserves Hodge filtration and bilinear pairing, but also the real structure.
	
	\appendix
	
	\section{Agmon Estimate}
	In this section, we let $V\in C^\infty(\R^n)$ be a non-negative function with finite isolated zeros. Moreover
	\[\lim_{|x|\to\infty}\frac{|\nabla V|}{(V+1)^{3/2}}(x)=0.\]
	
	The metric $\tilde{g}:=Vg_0$ is called Agmon metric with respect to $V$, where $g_0$ is the standard metric on $\R^n$. Let $\tilde{d}$ be the distance function induced by the Agmon metric, and $\rho(x):=\tilde{d}(x,0).$
	We summarize several nice properties of Agmon distance here (c.f. \cite{DY2020cohomology} )
	\begin{lem}\label{propagm}
		\begin{enumerate}
			\item $|\nabla \rho|^2=V$, almost everywhere;
			\item If $|\nabla f|\leq V$, then $|f(x)-f(y)|\leq \tilde{d}(x,y)$. In particular, if $f(0)=0$, $|f(x)|\leq \rho(x).$
		\end{enumerate}
	\end{lem}
	
	\begin{lem}\label{Agmonre1}
		\def\dvol{\mathrm{dvol}}
		
		Assume that $w\in L^2(\R^n),0\leq u\in L^2(M)$, such that $(\Delta +V)u\leq w$ outside a compact subset $K\subset M$ in the weak sense (where the interior of $K$ contains all the zeros of $V$). That is
		\[\int_{\R^n-K}\left(\nabla u\nabla v+V u v\right) \dvol\leq\int_{\R^n-K}w\cdot v~\dvol, \ \ \forall\ 0\leq v\in C_c^\infty(M-K).\]
		If $\int_{\R^n-K}|w|^2V^{-1}\exp(2b\rho)\dvol<\infty$ for some $b\in(0,1)$, then for any compact set $L$ such that $L^\circ\supset K$, one has
		\begin{align}\begin{split}\label{techeq}\int_{\R^n-L}V|u|^2\exp(2b\psi)\dvol&\leq C(b,K,L)  \int_{\R^n-K}|w|^2V^{-1}\exp(2b\psi)\dvol \\
		&+ C(b,K,L) \int_{L-K}V|u|^2\exp(2b\psi)\dvol. \end{split}\end{align}

	\end{lem}
	
	\begin{proof}
		This is exactly  \cite[Lemma 3.1]{DY2020cohomology}.
	\end{proof}

	First, let us recall the De Giorgi-Nash-Moser theorem:
	\begin{thm}\label{moserde}
		
		Let $B_r:=\{x\in\R^n:|x|<r\}.$ 
		Suppose that $u\in L^2(B_r)$, $w\in L^N(B_r)$ for some $N>n/2$, s.t. $\Delta u\leq w$ in the weak sense (and $u\geq 0.$).
		Then 
		\[\sup_{y\in B_r}u(y)\leq {C}(r^{-n/2}\|u\|_{L^2(B_{2r})}+r^{-n/N}\|w\|_{L^N(B_{2r})}),\]
		where $C>0$ is a constant depending on $n$ and $N.$
	\end{thm}
	\begin{proof}
		See  \cite[Theorem 4.1]{han2011elliptic} for a reference.
	\end{proof}
	
	\begin{lem}\label{epde12}
		Suppose $u,w$ satisfy the same conditions as in Lemma \ref{Agmonre1} for any compact set containing the zeros of $V$. Moreover, assume that
		$w$ satisfies
		\[\|w\|_{L^{N}_{wt}}^{N}:=\int_{M}|w|^{N}\exp(Nb\rho)dvol<\infty,\]
		\[\|w\|_{L^{2}_{wt}}^{2}:=\int_{M}|w|^{2}\exp(2b\rho)dvol<\infty.\]
		Then for $a\in(0,b)$,
		
		\[|u|\leq C(\psi,V,n,N,a,b)\left(\|u\|_{L^2}+\|w\|_{L^{N}_{wt}}+\|w\|_{L^{2}_{wt}}\right)\exp(-a\rho).\]
		
	\end{lem}
	
	\begin{proof}
		First, we fix a compact subset ${K}$, such that outside ${K}$, $\frac{|\nabla V|}{(V+1)^{3/2}}\leq\frac{(b-a)}{2},$
		and let $L:=\{x\in\R^n: \tilde{d}(x,{K})\leq3\}.$
		
		\def\tb{\tilde{B}}
		Denoted $\tb_r(x):=\{y\in\R^n:\tilde{d}(y,x)<r\}$.
		For $x_0\notin L$,  set $l=\sup_{x\in \tb_2(x_0)}V(x)$, and $r=1/(2l).$ Then one can easily verify that $B_{2r}(x_0)\subset \tb_1(x_0)$.
		
		Choose $y_0\in \overline{\tb_2(x_0)}$ so that $V(y_0)\in (l/2,l].$
		By Lemma \ref{Agmonre1} and De Giorgi-Nash-Moser estimate (Theorem \ref{moserde}),
		\begin{align}\begin{split}\label{add3}
		&|u(x_0)|^2e^{2b\rho(x_0)}\leq C(n,N)(r^{-n}\|u\|^2_{L^2(B_{2r}(x_0))}e^{2b\rho(x_0)}+r^{-2n/N}\|w\|^2_{L^N(B_{2r}(x_0))}e^{2b\rho(x_0)})\\
		&\leq{C(n,N,b)}\left({r^{-n}}\int_{\tb_1(x_0)}|u|^2(y)e^{2b\rho(y)}dy +r^{-2n/N}(\int_{\tilde{B}_1(x_0)}|w(y)|^Ne^{Nb\rho(y)}dy)^{2/N}\right)\\
		&\leq C(n,N,b,a,V)\left({r^{-n}}\int_{L-K}|u|^2e^{2b\rho}dy+r^{-n}\int_{\R^n-K}|w|^2e^{2b\rho}dy +r^{-2n/N}\|w\|_{L^N_{wt}}^2\right)\\
		&\leq C(n,N,b,a,V)\left(|V(y_0)|^n\|u\|_{L^2}^2+|V(y_0)|^{n}\|w\|_{L^2_{wt}}^2 +|V(y_0)|^{2n/N}\|w\|_{L^N_{wt}}^2\right).
		\end{split}\end{align}
Here the first inequality follows from the fact that 
\be\label{add1}
|\rho(y)-\rho(x_0)|\leq2.
\ee
		Proceeding as in \cite{DY2020cohomology}, one has \begin{equation}\label{add2}|V(y_0)|^2 \leq C(V,a,b)\exp((b-a) \rho(y_0)).\end{equation}
		Hence, by \eqref{add1}, \eqref{add2} and \eqref{add3}, the result follows for $x_0\notin L$
		
		For $x_0\in L$, we have classical De Giorgi-Nash-Moser estimate
		\[|u(x_0)|\leq C(a,b,V,N,n)(\|u\|_{L^2}+\|w\|_{L^N}). \]
		
		Since in $L$, $\exp(-a \rho)$ has an upper and a lower bound depending on $a,b$ and $V$, the result follows.
	\end{proof}
	
	\subsection{Witten deformation and Agmon estimate}\label{witagm}

	Suppose that $f$ is a non-degenerate homogeneous polynomial on $\C^n$, then for any $a\in(0,1)$, there exists $r_0:=r_0(a)>0$, s.t outside $B_{r_0}:=\{x\in \R^n:|x|\leq r_0\}$, the Witten Laplacian $\Delta_{2\Re(f)}\geq \Delta+a|2\nabla \Re(f)|^2$ (c.f. \cite{DY2020cohomology} and \cite{zhang2001lectures}), i.e. for any smooth form $w$,
	\[g_0(\Delta_{2\Re(f)}w,w)(p)\geq g_0((\Delta+a|2\nabla \Re(f)|^2)w,w)(p)\]
	for any point $p\notin B_{r_0}.$
	
	Then if $\Delta_{2\Re(f)} u=v$ for some differential forms $u$ and $v$, Bochner formula and Kato's
	inequality tells us that
	\be\label{boc}(\Delta+a|2\nabla \Re(f)|^2)|u|\leq |v|\ee
	weakly (outside $B_{r_0}$).
	
	In this case, let $\rho$ be the Agmon distance with respect to $V:=|\nabla 2\Re(f)|^2$. The Agmon estimate discussed in the previous section is applicable for Witten Laplacian.
	\section{Proof of Proposition \ref{riebil1}}
	It suffices to prove the case of $u=0$. 
	
	Recall that by our construction, $\gamma_k=\sup_{t>0}\Phi_t^*(\sigma_k)$ for some $\sigma_k\in H_{n-2}(V_{-1}).$ 
	Let \[\epsilon_1:=(2\sup_{z\in\cup_{k=0}^{\mu-1}\sigma_k}|z|)^{-n}.\]
	
	Let \[V:=\{z\in\C^n:|z|=1, |\Re(f)|\geq \epsilon_1\},\] and $U$ be the cone
	\[U:=\{z\in\C^n:\frac{z}{|z|}\in V\},\]
	then since $f$ is homogeneous of degree $n$, one has $\cup_{k=0}^{\mu-1}\gamma_k\subset U.$
	\def\tf{\tilde{f}}
	\begin{lem}\label{tff4}
		For every sufficiently small $\epsilon>0$, there exists a smooth function $\tf:\C^n\to \R$, such that
		\begin{itemize}
			\item $\tf\geq |\Re(f)|.$ Moreover, $\tf=|\Re(f)|$ in $U.$
			\item $|\nabla \tf|\leq (1+c(n,f)\epsilon)|\nabla \Re(f)|$ for some $c(n,f)>0.$ Hence by Proposition \ref{propagm},  we also have \be\label{tfrho}\tf\leq\frac{(1+c(n,f)\epsilon)}{2}\rho.\ee
			\item $\tf\geq \epsilon |z|^n$.
		\end{itemize}
	\end{lem}
	\begin{proof}
		
		Let $\eta\in C^\infty(\R)$, $\epsilon<\epsilon_1$, such that
		\begin{itemize}
			\item $\epsilon/2\leq\eta(x)\leq \epsilon$ if $|x|\leq\epsilon,$  and $\eta(x)=|x|$,  if $|x| \geq \epsilon$;
			\item $\eta\geq \epsilon/2$;
			\item $|\eta'|\leq 1$.
		\end{itemize}

		Let $(r,\theta)$ be the polar coordinates of $\C^n$, and $\nabla^{\theta}$ be gradient with respect to the standard metric on $S^{2n-1}:=\{z\in\C^n:|z|=1\}$. Now set $\tf(r,\theta):=r^n\eta\circ\Re(f)(1,\theta)$.

		By our construction, one can see that 
		\begin{itemize}
			\item \be\label{tff1}\mbox{$\epsilon|z|^n/2\leq\tf\leq\epsilon|z|^n$ if $z\notin U$, and $\tf(z)=|\Re(f)|(z)$ if $z\in U$}\ee
			\item \be\label{tff2}\tf\geq \epsilon |z|^n/2;\ee
			\item \be\label{tff}|\nabla^{\theta} \tf|\leq |\nabla^{\theta} \Re(f)|.\ee
		\end{itemize}

		Note that in $U$, $|\nabla \tf|= |\nabla \Re(f)|$.
		In polar coordinates, 
		\begin{align}\begin{split}\label{tff3}&|\nabla \Re(f)(r,\theta)|^2=n^2r^{2n-2}|\Re(f)(1,\theta)|^2+r^{2n-2}|\nabla^{\theta} \Re(f)(1,\theta)|^2,\\
		&|\nabla \tf(r,\theta)|^2=n^2r^{2n-2}\tf(1,\theta)+r^{2n-2}|\nabla^{\theta} \tf (1,\theta)|^2.\end{split} \end{align}
		
		Let $\epsilon_2:=\inf_{|z|=1}| \Re(f)|$, then by (\ref{tff1}), (\ref{tff}) and (\ref{tff3}), one can see easily that outside $U$,
		\[|\nabla \tf|\leq (1+\frac{\epsilon}{n\epsilon_2})|\nabla \Re(f)|.\]
	\end{proof}
	
	Notice that $e^{-\bar{f}}A_a$ is $d_{2\Re(f)}$ closed, although $e^{-\bar{f}}A_a$ is not $L^2$ integrable, one still have the following formulation of Hodge decomposition:
	\begin{lem}\label{decom}
		One has the following decomposition
		\[e^{-\bar{f}}A_a=w'+d_{2\Re(f)}\beta,\]
		where $w'$ is a harmonic form. Moreover, $\beta$ satisfies
		\begin{enumerate}[(a)]
			\item\label{estim1} $e^{f+\bar{f}}\beta$ has exponential decay on $U$, i.e., there exists $c,C>0$, such that $|e^{f+\bar{f}}\beta|\leq Ce^{-c|z|^n}$ in $U.$
			\item\label{estim2} There exist $a\in(0,1), C>0$, such that $|\beta|\leq Ce^{a\rho},$ where $\rho$ is the Agmon distance.
		\end{enumerate}
	\end{lem}
	Once we have the decomposition in the lemma, one can see 
	\be\label{int1}\int_{\gamma_k^-}e^{f}A_a=\int_{\gamma_k^-}e^{f+\bar{f}}w'\ee
	and for any harmonic form $w$,
	\be\label{int2}\int_{\C^n} e^{-\bar{f}}A_a\wedge*w=\int_{\C^n}w'\wedge*w.\ee
	(This is because, since $\beta$ satisfies the estimate above, by Stoke formula
	
	\[\int_{\gamma_k^-}d e^{f+\bar{f}}\beta=0,\]
	and integration by parts,
	\[\int_{\C^n}d_{2\Re(f)}\beta\wedge*w=\int_{\C^n}\beta\wedge* d_{2\Re(f)}^\dagger w=0.)\]
	
	By (\ref{int1}), (\ref{int2}) and Proposition \ref{riebil}, one obtains Proposition \ref{riebil1}.
	
	Now it suffices to prove Lemma \ref{decom}.
	\begin{proof}[Proof of Lemma \ref{decom}]
		We fix a function $\tf$ that satisfies the conditions in Lemma \ref{tff4} for a fixed $\epsilon<\min\{\epsilon_1,\frac{1}{16c(n,f)}\}$. 
		
		\begin{enumerate}[Step 1.]
			\item Let 
			$$d_{tw,0}:=e^{-2\Re(f)-3\tf/2}\circ d\circ e^{2\Re(f)+3\tf/2}=e^{-3\tf/2}\circ d_{2\Re(f)}\circ e^{3\tf/2},$$ 
			and 
			$$\Delta_{tw,0}:=d_{tw,0}d_{tw,0}^{\dagger}+d_{tw,0}^{\dagger}d_{tw,0}$$ be the Witten Laplacian with respect to $d_{tw,0}$. Then one can see that $e^{-3\tf/2-\bar{f}}A_a\in L^2(\C^n)$ (since $|e^{-3\tf/2-\bar{f}}A_a|\leq Ce^{-\tf/2}|z|^{l_a}\leq Ce^{-\epsilon|z|^n/2}|z|^{l_a}$), 
			and \[d_{tw}e^{-3\tf/2-\bar{f}}A_a=0.\] As a reselt, we have Hodge decomposition (c.f. \cite{DY2020cohomology})
			\[e^{-3\tf/2-\bar{f}}A_a={w}_0+d_{tw,0}\beta_0,\]
			where $w_0$ is $\Delta_{tw,0}$-harmonic, and we can also assume that $\beta_0\in \Im(d_{tw,0}^{\dagger}).$
			
			Since $\left|\nabla \left(2\Re(f)+3\tf/2\right)\right|\geq 2|\nabla \Re(f)|-3/2|\nabla \tf|\geq 13/32|\nabla \Re(f)|$, Agmon estimate tells us that
			\be\label{wzero}|w_0|\leq Ce^{-3\rho/16}\ee
			for some $C>0$.
			
			Moreover, by our choice of $\beta_0$, one has
			\[d_{tw,0}^{\dagger}e^{-3\tf/2-\bar{f}}A_a=d_{tw,0}^{\dagger}\tilde{w}_0+d_{tw,0}^{\dagger}d_{tw,0}\beta_0=d_{tw,0}^{\dagger}d_{tw,0}\beta_0=\Delta_{tw,0}\beta_0.\]

			Let $\epsilon_3:=n\sup_{|z|=1}|\nabla 2\Re(f)|$, then $\rho\leq \epsilon_3|z|^n$. Hence, there exists $b>0$, such that
			$e^{b\rho}d_{tw,0}^{\dagger}e^{-3\tf/2-\bar{f}}A_a$ is $L^2$ and $L^N$ integrable for some $N>n.$
			Hence, by Agmon estimate again, there exists $c<b$, such that
			\be|\beta_0|\leq Ce^{-c\rho}.\ee
			
			Now let $\tilde{\beta}_0:=e^{3\tf/2}\beta_0$, $\tilde{w}_0=e^{3\tilde{f}/2}w_0$, then
			$e^{-\bar{f}}A_a=\tilde{w}_0+d_{2\Re(f)}\tilde{\beta}_0.$
			Moreover, $\tilde{w}_0$ is $d_{2\Re(f)}$-closed. In $U$, $\tilde{f}=|\Re(f)|$, hence for some $c'>0$, \[|e^{f+\bar{f}}\tilde{\beta}_0|\leq Ce^{-\tf/2-c\rho}\leq Ce^{-c'|z|^n}.\] 
			In addition, \[|\tilde\beta_0|\leq Ce^{3\tf/2-c\rho}\leq Ce^{(3(1+\epsilon c(n,f))/4-c)\rho}.\] Hence $\tilde{\beta}_0$ satisfies the estimate (\ref{estim1}) and (\ref{estim2}) above.
			
			\item Although $\tilde{w}_0$ is $d_{2\Re(f)}$-closed, it is not harmonic. Also, $\tilde{w}_0$ may not be $L^2$ integrable. To continue, we will use the techniques similar to those used in \cite[\S 7.5]{DY2020cohomology}.
			Let $$d_{tw,1}=e^{-\tf/4}\circ d_{tw,0}\circ e^{\tf/4}=e^{-5/4\tf-2\Re(f)}\circ d\circ e^{5/4\tf+2\Re(f)},$$ and $\Delta_{tw,1}$ be the Witten Laplacian with respect to $d_{tw,1}.$
			
			By (\ref{wzero}) and (\ref{tfrho}), $\alpha_1:=e^{\tf/4}w_0$ is $d_{tw,1}$ closed and $L^2$-integrable. Hence, we have Hodge decomposition 
			\[\alpha_1=w_1+d_{tw,1}\beta_1,\]
			where $w_1$ is $\Delta_{tw,1}$ harmonic, and we may assume $\beta_1\in \Im(d_{tw,1}^{\dagger})$.
			
			Since 
			\begin{align*}
			\left|\nabla \left(5/4\tf+2\Re(f)\right)\right|&\geq 2|\nabla \Re(f)|-5/4|\nabla \tf|\\
			&\geq2|\nabla \Re(f)|-3/2|\nabla \tf|\geq 13/32|\nabla \Re(f)|, 
			\end{align*}
			we also have $|w_1|\leq C e^{-3\rho/16}.$ Similarly, $|\beta_1|\leq Ce^{-c\rho}$ for some $c>0.$
			
			Then, let $\tilde{w}_1=e^{5\tf/4}w_1,$ $\tilde{\beta}_1=e^{5\tf/4}\beta_1$, then we have
			\[\tilde{w}_0=\tilde{w}_1+d_{2\Re(f)}\tilde{\beta}_1.\]
			Here $\tilde{\beta}_1$ satisfies the estimates stated above, and $\tilde{w}_1$ is $d_{2\Re(f)}$ closed (but may not be $L^2$ integrable again).
			\item Now let $d_{tw,k}=e^{-k\tf/4}\circ d_{tw,0}\circ e^{k\tf/4}$. Repeating the arguments in Step 2 for 6 times, eventually we get
			$\alpha_6=e^{\tilde{f}/4} w_5$ is $d_{tw,6}(=d_{2\Re(f)})$ closed and $L^2$ integrable. Hence we have Hodge decomposition
			\[\alpha_6=w_6+d_{2\Re(f)}\beta_6,\] where $w_6$ is $\Delta_{2\Re(f)}$ harmonic, and $\beta_6$ satisfies the estimates stated above.
		\end{enumerate}
		Eventually, set $w'=w_6$, $\beta:=\tilde{\beta}_0+\cdots+\tilde{\beta}_5+\beta_6$, we finish the proof.
	\end{proof}
	
	\section{Modified Real Structure on $\Jac(F)$}\label{last}
	
	Let ${\phi_a}$ be a monomial representative of a basis for $\Jac(f)'$, such that the degree of $\phi_a$ is increasing. When $|u|$ is small, ${\phi_a}$ also serves as a monomial representative of a basis for $\Jac(F)'$.
	
	Additionally, $\Jac(F)'$ has a pole order filtration given by
	\[F^p:=Span\{\phi\in \Jac(F)'| \frac{deg(\phi)}{n}\in\mathbb{N},\frac{deg(\phi)}{n}\leq n-2-p\},\]
	as well as a Grothendieck residue pairing $Q$ given by
	$$Q(\phi_1,\phi_2)=\frac{\tilde{c}_{ab}}{(2\pi i)^n}\int_{\Gamma(\epsilon)}\frac{\phi_1\phi_2 dz_1\wedge\cdots\wedge dz_n}{\p_1F(\cdot,u)\cdots\p_nF(\cdot,u)},\forall \phi_1,\phi_2\in \Jac(F),$$
	where $\tilde{c}_{ab}$ is defined as $\tilde{c}_{ab}=\frac{(-1)^{a(a-1)/2+b(b-1)/2+n+(b-1)^2}}{(a-1)!(b-1)!}$, and $\Gamma(\epsilon)=\{|\p_iF|=\epsilon: i=1,\ldots,n\}$.

	\def\bp{\tilde{\phi}}
	Let $(\T_{ab})$ be the transition matrix derived in Theorem \ref{lgcy1}, and define $$\bp_a:=\phi_a-\sum_{b\neq a}\frac{(-1)^{n_b-1}(n_b-1)!}{(-1)^{n_a-1}(n_a-1)!}\T_{ab}\phi_b.$$ Let $w_a$ be the harmonic form w.r.t. $\phi_a$. Since $w_a$ is a basis, we have $\bar{w}_a:=\sum_bM_{\bar{a}b} w_b$. We then define the real structure $\tau$ on $\Jac(F)'$ as follows:
	\[\tau(\bp_a):=\sum_{b}M_{\bar{a}b}\bp_b.\]

	From the construction of $r$ and Theorem \ref{134}, we know that the real structure $\tau$ is preserved by $R$, meaning that for any $\phi\in \Jac(F)$, we have $R_u(\tau(\phi))=\overline{R_u(\phi)}$.
	
	Moreover, because $\T$ is lower triangular and Theorem \ref{rpqdegree} holds, we know that $\tau$ is compatible with the pole order filtration on $\Jac(F)'$. That is, for all $p$, we have $\Jac(F)'\cong F^p\oplus \tau(F^{n-1-p})$.

	It follows from the fact that the diagonal of $\T$ is $1$, Theorem \ref{fanlan}, and \cite[Theorem 3.4]{shen2016explicit} that 
	\begin{equation}\label{better}Q(R_u(\phi_1),R_u(\tau(\phi_2)))=\int_{X_F}R_u(\phi_1)\wedge \overline{R_u(\phi_2)}.\end{equation}
	\def\tt{\tilde{\tau}}
	\begin{rem}
		In \cite[(4.2)]{CECOTTI1991N}, the real structure $\tt$ is simply given by 
		\[\tt(\phi_a):=\sum_{b}M_{\bar{a}b}\phi_b.\]
		While $\tt$ is compatible with the pole order filtration, it is not preserved under the map $R$.
		
		In \cite{cecotti1991topological}, Cecotti-Vafa provide a less refined structure for the matrix $\T$ and show that (see \cite[(A.10)]{cecotti1991topological})
		\[Q(R_u(1),R_u(\tt(1)))=c(u,\bar{u})\int_{X_F}R_u(1)\wedge \overline{R_u(1)}\]
		for a positive function $c(u,\bar{u})$. However, with a better understanding of $\T$, we have (\ref{better}).
	\end{rem}

	\bibliography{lib}
	\bibliographystyle{plain}
	
\end{document}